\documentclass{article}
\usepackage[utf8]{inputenc}
\usepackage[round,authoryear]{natbib}
\usepackage[margin=1in]{geometry}
\usepackage{amsmath}
\usepackage{amssymb}
\usepackage{amsthm}
\usepackage{mathrsfs}
\usepackage[colorlinks=true,urlcolor=blue,linkcolor=red,citecolor=blue]{hyperref}%
\usepackage{mathtools}
\usepackage{graphicx}
\usepackage{subcaption}
\usepackage{array}
\usepackage{pifont}
\usepackage{algorithm}
\usepackage{algpseudocode}
\newcolumntype{C}[1]{>{\centering\arraybackslash}m{#1}}

\def\argmin{\mathop{\rm argmin}\nolimits}

\newcommand{\bp}{\boldsymbol{p}}
\newcommand{\bq}{\boldsymbol{q}}

\newcommand{\bw}{\boldsymbol{w}}
\newcommand{\bx}{\boldsymbol{x}}

\newcommand{\btheta}{\boldsymbol{\theta}}

\usepackage{algorithm}
\usepackage{algpseudocode}
\newtheorem{thm}{Theorem}
\newtheorem{lem}{Lemma}
\newtheorem{cor}{Corollary}
\usepackage[blocks]{authblk}
\usepackage[T1]{fontenc}
\usepackage{babel}
\newtheorem{prop}{Property}
\newtheorem{defn}{Definition}
\newtheorem{ex}{Example}
\newtheorem{rem}{Remark}

\title{\textit{t}-Entropy: A New Measure of Uncertainty with Some Applications\thanks{To appear at the IEEE Symposium on Information Theory (ISIT), 2021.}}
\author[1]{Saptarshi Chakraborty\thanks{Joint first authors contributed equally to this work.}}
\author[2]{Debolina Paul$^\dagger$}
 \author[3]{Swagatam Das\thanks{Correspondence to: \href{mailto:swagatam.das@isical.ac.in}{swagatam.das@isical.ac.in}}}
 \affil[1]{Department of Statistics, University of California, Berkeley}
 \affil[2]{Indian Statistical Institute, Kolkata, India}
 \affil[3]{Electronics and Communication Sciences Unit, Indian Statistical Institute, Kolkata, India}

\date{\vspace{-5ex}}
\begin{document}

\maketitle

\begin{abstract}
   The concept of Entropy plays a key role in Information Theory, Statistics, and  Machine Learning. This paper introduces a new entropy measure, called the \textit{t}-entropy, which exploits the concavity of the inverse-tan function. We analytically show that the proposed \textit{t}-entropy satisfies the prominent axiomatic properties of an entropy measure. We demonstrate an application of the proposed entropy measure for multi-level thresholding of images. We also propose the entropic-loss as a measure of the divergence between two probability distributions, which leads to robust estimators in the context of parametric statistical inference. The consistency and asymptotic breakdown point of the proposed estimator are mathematically analyzed.  Finally, we also show an application of the \textit{t}-entropy to feature weighted data clustering.
\end{abstract}

\section{Introduction}
 The concept of entropy is a very fundamental tool in statistical mechanics, thermodynamics, information sciences, and statistics. In physics, entropy typically refers to the measure of randomness in a physical system. In thermodynamics, it is interpreted as the amount of molecular disorder within a macroscopic system. The \textit{second law of thermodynamics} states that the entropy of an isolated system will never decrease over time. The system spontaneously evolves towards a thermodynamic equilibrium, where it attains its maximum entropy, i.e. a state of maximum disorder. In this paper, we will focus on defining a new entropy function as a measure of uncertainty in an information system.
   
\subsection{A Brief History of Entropy in Information Theory}
In Information Theory, Claude Shannon \citep{shannon1948mathematical} is known as the first to introduce a measure of randomness or uncertainty in a discrete distribution in 1948. Suppose $X$ be a discrete random variable, which takes values in $\mathcal{X}=\{x_1,\dots,x_n\}$ with $P(X=x_i)=p_i$. Shannon's proposed measure, known as Shannon's entropy, is given by $$H_{Shannon}(X)=-\sum_ip_i \log(p_i),$$ where $p_i$'s are the probabilities associated with various realizations of $X$. Shannon's entropy has various interesting properties such as non-negativity, attaining maximum when $p_i$'s are all equal, equals $0$ when the distribution is degenerate and is additive.
\par
Shannon's entropy can also be viewed as the average information contained in a distribution. Consider any point $x \in \mathcal{X}$. If $P(X=x)$ is very small, then the chance of obtaining the value $x$ of the random variable $X$ is also very small. The occurrence of a small probability event contains more information than the occurrence of a large probability event (which is more certain to occur). Thus, the information for an event $\{X=x\}$ should be an increasing function of $\frac{1}{P(X=x)}$. In order to convert this information to bits, Shannon proposed $I(x)=\log_2\big(\frac{1}{P(X=x)}\big)$ as a measure of information of observing $X=x$. Shannon's entropy thus boils down to the average information contained in the random variable i.e. $\mathbb{E}_{X \sim P} [I(X)]=\sum_{i=1}^n p_i \log\big(\frac{1}{p_i}\big)=H_{Shannon}(X)$. 
\par
The extension of Shannon's Entropy for continuous random variables is known as Differential Entropy \citep{cover2012elements}. It is defined as: $$DE(X)=-\int f(x)\log\big(f(x)\big)dx,$$ where $f(x)$ is pdf of the random variable $X$.\par
In 1961, Alfred R\'enyi proposed a generalization of Shannon's entropy, which is known as the R\'enyi entropy \citep{renyi1961measures}. It is defined as: 
$$H_{\alpha}(p_1,p_2,\dots,p_n)=-\frac{1}{1-\alpha}\log\bigg(\sum_{k=1}^n p_k^{\alpha}\bigg),$$
where $\alpha >0$ and $\alpha \neq 1$. It is a generalization in the sense that for $\alpha \rightarrow 1$, the R\'enyi entropy converges to Shannon's entropy.\par
Another famous entropy measure is Tsallis entropy \citep{tsallis1988possible} proposed by Constantino Tsallis in 1988. This measure is a generalization of the standard Boltzmann-Gibbs entropy \citep{jaynes1965gibbs}. Tsallis entropy is defined as: $$S_q(p)=\frac{\kappa}{q-1}(1-\sum_i p_i^q),$$
where, $\kappa$ is the Boltzman's constant and $q$ is a parameter. As we take the limit $q\rightarrow 1$, Tsallis entropy becomes Boltzmann-Gibbs entropy, which is nothing but a constant multiple of the Shannon's entropy.
\par
Other entropies that are frequently used in information theory are Sharma-Mittal entropy \citep{sharma1975new}, Cumulative Residual Entropy (CRE) \citep{rao2004cumulative}, Havrda and Chavrat entropy \citep{havrda1967quantification}, Awad entropy \citep{awad1987application}, and their extensions.
\par
Shannon's entropy became important in quantifying randomness present in diverse scientific fields such as financial analysis \citep{sharpe1998investments}, data compression \citep{salomon2007concise}, statistics \citep{kullback1959statistics}, and information theory \citep{cover1991entropy}. Other entropy measures can also be used in these applications.

\begin{table*}[ht]
\centering
\caption{Some Standard Entropy Measures in Literature Along with the Parameter Values}
\label{entropy}
\resizebox{\textwidth}{!}{
\renewcommand{\arraystretch}{2}
\begin{tabular}{|l|c|c|c|}
 \hline
 \textbf{Entropy} & \textbf{Formula} & \textbf{Parameters} & \textbf{Parameter Space} \\
 \hline
 1.Shannon's Entropy \citep{shannon1948mathematical} &  $-\sum_ip_i\log(p_i)$  & - & - \\
 \hline
 2.Boltzmann-Gibbs Entropy \citep{jaynes1965gibbs} & $-\kappa\sum_ip_i\log(p_i)$ & -& - \\
 \hline
 3.Differential Entropy\citep{cover2012elements} &  $-\int f(x)\log\big(f(x)\big)dx$  & - & - \\
 \hline
 4.R\'enyi Entropy \citep{renyi1961measures} & $-\frac{1}{1-\alpha}\log\bigg(\sum_{k=1}^n p_k^{\alpha}\bigg)$ & $\alpha$& $\alpha >0$, $\alpha \neq 1$\\
 \hline
 5.Tsallis Entropy \citep{tsallis1988possible} & $\frac{\kappa}{q-1}\big(1-\sum_i p_i^q\big)$ & $q$ & $q>0$, $q \neq 1$\\
 \hline
 6. Sharma-Mittal Entropy \citep{sharma1975new} &$\frac{1}{1-\beta}\bigg(\big(\int p(x)^\alpha dx\big)^{\frac{1-\beta}{1-\alpha}}-1\bigg)$ & $\alpha,\beta$ & $\alpha,\beta>0$, $\alpha \neq 1 $, $\beta \neq 1$, $\alpha \neq \beta$ \\
 \hline
 7. Cumulative Residual Entropy (CRE) \citep{rao2004cumulative} & $-\int_0^\infty p(|X|>x)\log p(|X|>x)dx$ & - & -\\
 \hline
 8. Havrda and Chavrat Entropy \citep{havrda1967quantification} & $\frac{1}{2^{(1-\alpha)}-1}\bigg(\int p(x)^\alpha dx -1\bigg)$ & $\alpha$ & $\alpha>0$\\
 \hline
 9. Awad Entropy \citep{awad1987application} & $-\int f(x)\log\bigg(\frac{f(x)}{\sup_xf(x)}\bigg)dx$ & - & -\\
 \hline
 10.$t$-Entropy (proposed) & $\sum_{i} p_i \tan^{-1}\bigg(\frac{1}{p_i^c}\bigg) - \frac{\pi}{4}$ & $c$ & $c>0$\\
 \hline
\end{tabular}
}
\end{table*}

The application of entropy can be found in almost all corners of modern machine learning, from optimal transport to neural networks. In optimal transport, the computation of the Kantorovich distance \citep{villani2008optimal} requires solving of a linear program, which can be computationally intensive. The introduction of an entropy-based regularizer results in fixed-point iteration \citep{cuturi2013sinkhorn}, which is generally faster than the linear program. The application of entropic regularizers can also be found in semi-supervised learning \citep{grandvalet2005semi,audiffren2015maximum} and clustering \citep{jing2007entropy,chakraborty2020entropy,paul2020bayesian,chakraborty2020automated}. Entropy has been traditionally used in decision trees \cite{wang1984analysis} as an impurity measure for the nodes. Table \ref{entropy} discusses some of the standard entropies used in literature along with their parameter values and puts the proposed entropy in context. 
\paragraph{Motivation} As we have already discussed, Shannon's entropy can be viewed as the average information contained in the random variable. The information in the occurrence of an event $A$ is defined as $\log \big(\frac{1}{P(A)}\big)$. Despite the usefulness and interpretability of the $\log(\cdot)$, we note that it is unbounded and is very unstable near the value 0. We argue that information in an event should not only be finite but also should be bounded since one cannot hope to obtain infinite information by observing trials of a random variable, which is on finite support. To define a new entropy, one has to satisfy all the axiomatic requirements given by Shannon and Khinchin (see section \ref{axiom} for more details). For this purpose, we will define the information contained in an event $A$ to be $g\big(\frac{1}{P(A)}\big)$, where, $g(\cdot)$ is bounded and concave. Moreover, the domain of definition of $g(\cdot)$ must be the entire positive real line. A function which satisfies all the aforementioned properties is $\tan^{-1}(\cdot)$. We also know that the information contained in a probability one event is zero. In order to incorporate that we define our information as 
\begin{equation}
    I(A)=\tan^{-1}\bigg(\frac{1}{P(A)}\bigg)-\frac{\pi}{4}.
\end{equation}
The entropy, which is defined as the average information, becomes $H_c(X)=\mathbb{E}_{X \sim P} I(X)=\sum_{i=1}^n p_i \tan^{-1}\Big(\frac{1}{p_i}\Big)-\frac{\pi}{4}$, where, $\sum_{i=1}^np_i=1$. In order to generalize this further, we define the $t$-entropy for a probability vector $\bp$ as follows. 
\begin{equation}
    H_c(\bp)=\sum_{i=1}^n p_i \tan^{-1}\bigg(\frac{1}{p_i^c}\bigg) - \frac{\pi}{4},
\end{equation}
where $c$ is a positive constant. 

In what follows we summarize our main contributions:
\begin{itemize}
    \item We propose a new entropy with the notion of a new measure of information, which increases with increase in the amount of information and becomes saturated once the full information is known.
    \item We show analytically show that our proposed entropy satisfies all the prominent axioms of an entropy measure.
    \item Through extensive experiments, we give an application of our proposed entropy in the context of Image Segmentation. We show that the algorithms perform significantly better in the context of our proposed entropy compared to other existing ones. 
    \item We also provide an entropic-loss-based divergence and propose an estimator based on this divergence. We theoretically prove the consistency property of this estimator and also explore robustness of the same.
    \item This entropy is incorporated with the Entropy Weighted $k$-Means clustering formulation by \cite{jing2007entropy} and is shown to have superior performance in terms of standard cluster validation indices on benchmark datasets. All the relevant codes used in this paper can be downloaded from \href{https://github.com/DebolinaPaul/t-entropy}{https://github.com/DebolinaPaul/t-entropy}.
\end{itemize}
\par
\section{Background}
\subsection{Probability Spaces and Random Variables}\label{prob}
In this paper, we consider a finite probability space $(\Omega,\mathcal{F},P)$. Here, $\Omega$ is a finite set and $\mathcal{F}$ is the power set of $\Omega$, which gives us a $\sigma$-algebra. $P$ is a probability function $P: \mathcal{F} \to [0,1]$  defined on it. In this context, we can define a random variable as a function from $\Omega$ to $\mathbb{R}$, i.e. $X:\Omega \to \mathbb{R}$. For any set $A \subset \mathbb{R}$, one can define $\{X \in A\}:=\{\omega \in \Omega: X(\omega) \in A\}$. The distribution of $X$ is written as $P_X$. Note that $P_X: \mathbb{R} \to [0,1]$ such that $P_X(x)=P[X=x]$. For two random variables $X$ and $Y$ defined on the same probability space $(\Omega,\mathcal{F},P)$, the joint distribution function of $X$ and $Y$ is defined by the function $P_{XY}:\mathbb{R}^2 \to [0,1]$ such that, $P_{XY}(x,y)=P[\{X=x\} \cap \{Y=y\}]$. This definition can be similarly extended for more than two random variables. Also, we write $P(X|\Lambda)(x) = P[X =x | \Lambda]$ and $P_{X|Y} (x|y) = P_{X|Y=y} (x) = P[X =x | Y = y]$ for the respective conditional distributions (conditioned on an event $\Lambda$ and a random variable $Y$). More details about probability spaces and random variables can be found in \cite{gut2013probability}.
In a more general context, we define a distribution on a finite set $\mathcal{X}$ to be a function $p:\mathcal{X} \to [0,1]$ such that $\sum_{x \in \mathcal{X}}p(x)=1.$ 
\subsection{Axiomatic Definition and Properties}\label{axiom}
We take the axiomatic approach for defining the term ``entropy'' \citep{khinchin2013mathematical}. Let $X$ be a discrete random variable, taking $n$ distinct values. Without loss of generality, we may assume that these are the integers $1, 2,\dots, n$. Let us use some standard notations and abbreviations. We denote $P(X=i)$ by $p_i$. We want to represent the randomness of within this distribution to be represented as a single number $H(X)$, which we will call \textit{entropy} of $X$. We define, by way of abbreviation, the joint entropy of a two-component random variable $(X,Y)$ by $H(X,Y)$ and the entropy of the conditional distribution $P(Y|X=x)$ by $H(Y|X=x)$.
\par
The axioms as referred to in \citep{khinchin2013mathematical,nambiar1992axiomatic,chakrabarti2005shannon} are as follows:
\begin{enumerate}
    \item $H(X)$ depends only on the probability distribution of $X$, i.e. we can
change the labels of the events as much as we like without changing the value of the entropy.
    \item For a given $n$, $H(X)$ is maximal, when $p_i=\frac{1}{n}\forall i\in \{1,\dots,n\}$, i.e. the discrete uniform distribution has maximal entropy.
    \item $H_{n+1}(p_1,\dots,p_n,0)=H_{n}(p_1,\dots,p_n) \forall n\geq 1$, i.e., event of probability zero does not contribute to the entropy.
    \item $H(X,Y)\leq H(X)+H(Y|X)$, which is called the subadditivity property. 
\end{enumerate}
These are known as the Khinchin's Axioms \citep{khinchin2013mathematical,suyari2004generalization} for entropy. 

\section{Definition and Properties of the \textit{t}-Entropy}
\label{S2}
In this section, we formally define the $t$-entropy. We also state and prove some of its properties and show that it satisfies all the axioms of an entropic function, establishing that $t$-entropy is indeed a valid entropy.
\subsection{Formulation of the new entropy}\label{formulation}
We first define $t$-entropy for a probability vector $\bp$ in definition \ref{d:1}. We subsequently extend this definition to finite valued random variable in definition \ref{d:2}. The joint entropy and conditional entropy for two finite valued random variables are defined in definitions \ref{d:3} and \ref{d:4} respectively.  
\begin{defn}
\label{d:1}
Let $\bp=(p_1,\dots,p_n)$ be a probability vector defined on the set $\mathcal{X}=\{x_1,\dots,x_n\}$. The t-entropy is defined as 
$$H_c(\bp)=\sum_{i=1}^n p_i \tan^{-1}\bigg(\frac{1}{p_i^c}\bigg) - \frac{\pi}{4}.$$
\end{defn}

We will now define entropy corresponding to a random variable (taking values in a finite set).
\begin{defn}
\label{d:2}
Let $X$ be a random variable taking values in a finite set $\mathcal{X}$. Then the entropy of $X$ is defined by
$$H_c(X)=\sum_{x \in \mathcal{X}} p(x) \tan^{-1}\bigg(\frac{1}{p(x)^c}\bigg) - \frac{\pi}{4}.$$
\end{defn}
Similarly we define the joint entropy of two random variables as follows.
\begin{defn}
\label{d:3}
Let $X$ and $Y$ be two random variables taking values $\mathcal{X}$ and $\mathcal{Y}$, which are both finite sets. Then joint entropy of $X$ and $Y$ is defined by
$$H_c(X,Y)=\sum_{x \in \mathcal{X}}\sum_{y \in \mathcal{Y}} p(x,y) \tan^{-1}\bigg(\frac{1}{p(x,y)^c}\bigg) - \frac{\pi}{4}.$$
\end{defn}

The conditional entropy of two random variables is defined as follows.
\begin{defn}
\label{d:4}
Let $X$ and $Y$ be two random variables taking values $\mathcal{X}$ and $\mathcal{Y}$, which are both finite sets. Then conditional entropy of $X$ given $Y$ is defined by
$$H_c(X|Y)=\sum_{x \in \mathcal{X}}\sum_{y \in \mathcal{Y}} p(x,y) \tan^{-1}\bigg(\frac{1}{p(x|y)^c}\bigg) - \frac{\pi}{4}.$$
\end{defn}

\begin{defn}
\label{cor1}
Let $X$ and $Y$ be two random variables taking values $\mathcal{X}$ and $\mathcal{Y}$, which are both finite sets. The entropy for the random variable $X|Y=y$ is given by 
$$H_c(X|Y=y)=\sum_{x \in \mathcal{X}}p(x|y) \tan^{-1}\bigg(\frac{1}{p(x|y)^c}\bigg) - \frac{\pi}{4}$$
\end{defn}
\begin{figure}
    \centering
    \includegraphics[height=0.3\textwidth,width=0.5\textwidth]{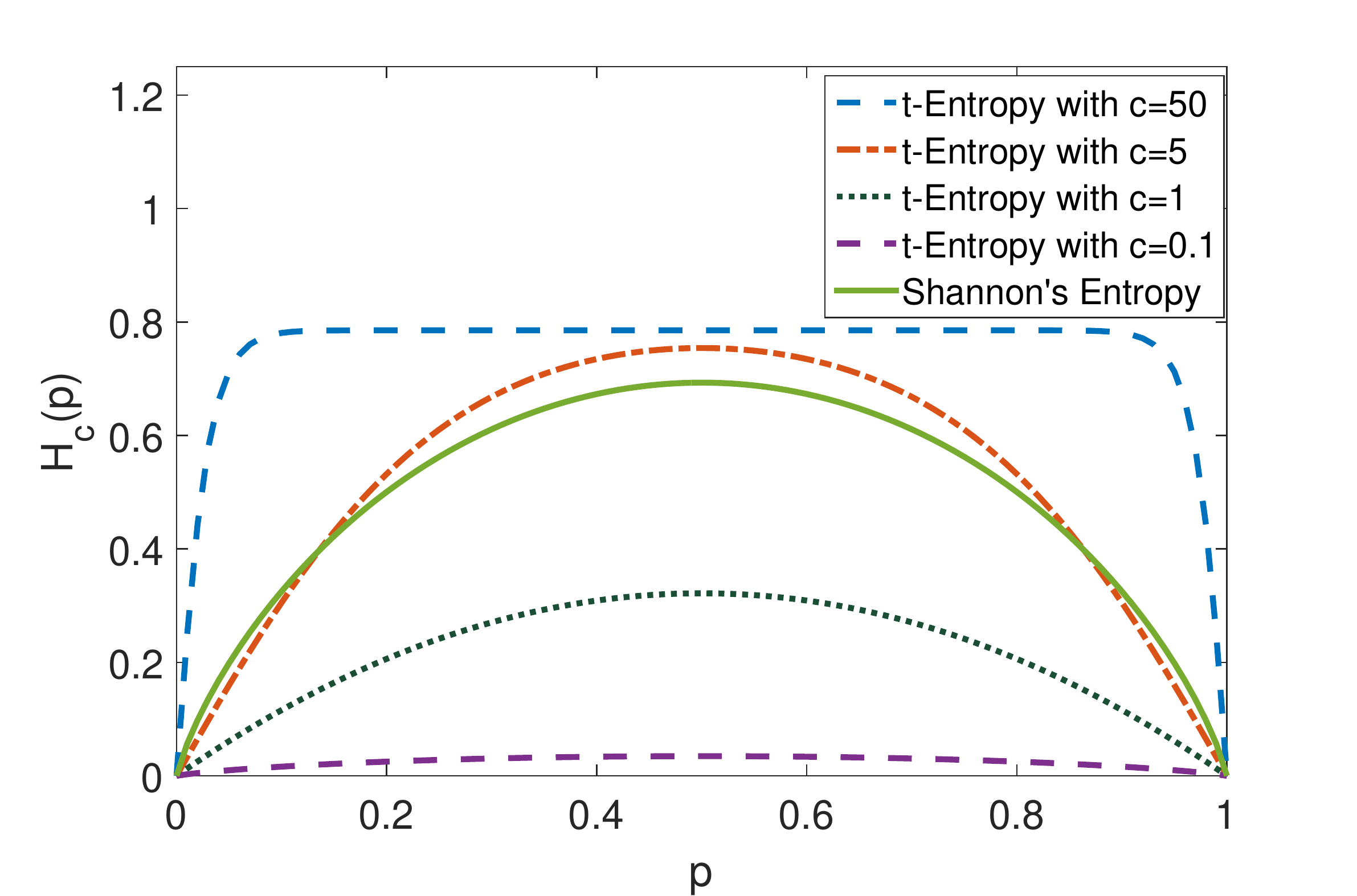}
    \caption{The plot of $H_c(p)$ for various values of $c$.}
    \label{fig:e_plot}
\end{figure}

\begin{ex}
\normalfont
Let $X$ be a Bernoulli random variable with parameter $p$. Then, the $t$-Entropy of $X$ is given by $H_c(p)=H_c(X)=p \tan^{-1}\bigg(\frac{1}{p^c}\bigg)+(1-p) \tan^{-1}\bigg(\frac{1}{(1-p)^c}\bigg)$. In Fig. \ref{fig:e_plot}, we plot $H_c(p)=H_c(X)$ against $p$ for various values of $c$. It can be easily seen from Fig. \ref{fig:e_plot} that $H_c(p)$ attains its maxima at $p=\frac{1}{2}$ and minima at the boundary points $p=0,1$. Also note that as $c$ increases, $H_c(p)$ also increases for $p \in (0,1)$ and in limit approaches $\frac{\pi}{4}$ except for the points $p=0,1$, where $H_c(p)=0$ for all $c>0$.
\end{ex}
\subsection{Properties of the proposed entropy function}\label{property}
In this section, we discuss some of the properties of $t$-entropy. Before we proceed, let us first state the following lemmas. The proof of all the lemmas are given in the Appendix \ref{a1}.

\begin{lem}
\label{l:1}
The function $f(x)=\tan^{-1}\left(\frac{1}{x}\right)$ is convex on $[0,\infty)$.
\end{lem}

\begin{lem}
\label{l:2}
For any $c>0$, the function $f(x)=x \tan^{-1}(\frac{1}{x^c})$ is concave on $[0,1]$. 
\end{lem}

We will first prove the non-negativity property of the proposed $t$-entropy.
\begin{prop}
\label{p:1}
(Non-negativity) Let $\bp=(p_1,\dots,p_n)$ be a probability vector defined on the set $\mathcal{X}=\{x_1,\dots,x_n\}$, then $H_c(P) \geq 0$.
\end{prop}
\begin{proof}

\begin{align*}
H_c(\bp) & =\sum_{i=1}^n p_i \tan^{-1}\bigg(\frac{1}{p_i^c}\bigg) - \frac{\pi}{4}\\
& = \sum_{i=1}^n p_i \bigg[\tan^{-1}\bigg(\frac{1}{p_i^c}\bigg) - \frac{\pi}{4}\bigg]\\
& \geq 0.
\end{align*}

The last inequality follows from the fact that the function $f(x)=\tan^{-1}(x)$ is an increasing function of $x$ in $[0, \infty)$ and for all $i \in \{1,\dots,n\}$, $p_i \leq 1 \implies \frac{1}{p_i^c} \geq 1 \implies \tan^{-1}\bigg(\frac{1}{p_i^c}\bigg) \geq \tan^{-1}(1)=\frac{\pi}{4}.  $
\end{proof}

We will now prove that $t$-entropy is continuous, so that changing the values of the probabilities by a small amount change the entropy by a small amount.
\begin{prop}
(Continuity) For any probability vector $\bp$, The function $H_c(\bp)$ is a continuous function of $\bp$.
\end{prop}
\begin{proof}
The result easily follows from the continuity of the function $f(x)=x\tan^{-1}(x)$ on $[0,1]$.
\end{proof}

Property \ref{p:sym} tell us that the $t$-entropy remains unchanged if the outcomes are reordered. This proves axiom (1) of Section \ref{axiom}.
\begin{prop}
\label{p:sym}
(Symmetry) For any probability vector $\bp=(p_1,\dots,p_n)$ and any permutation $\sigma:\{1,\dots,n\} \to \{1,\dots,n\}$, $H_c(\bp)=H_c(\sigma(\bp))$.
\end{prop}
\begin{proof}
We have,
\begin{align*}
    H_c(\sigma(\bp)) & =\sum_{i=1}^n \sigma(p_i) \tan^{-1}\bigg(\frac{1}{\sigma(p_i)^c}\bigg) -\frac{\pi}{4}\\
    & =\sum_{i=1}^n p_i \tan^{-1}\bigg(\frac{1}{p_i^c}\bigg)-\frac{\pi}{4}\\
    & = H_c(\bp).
\end{align*}
\end{proof}

We will now explore an interesting property of the $t$-entropy. Property \ref{p:2} states that for any $c>0$, the $t$-Entropy is bounded above by $\frac{\pi}{4}.$ 

\begin{prop}
 \label{p:2}
 (Boundedness) Let $\bp=(p_1,\dots,p_n)$ be a probability vector defined on the set $\mathcal{X}=\{x_1,\dots,x_n\}$, for any $n \in \mathbb{N}$ and any $c>0$, $H_{c}(\bp) \leq \frac{\pi}{4}$.
\end{prop}
\begin{proof}
We have, $H_c(\bp)=\sum_{i=1}^n p_i \tan^{-1}\bigg(\frac{1}{p_i}\bigg) - \frac{\pi}{4} \leq  \frac{\pi}{2}\sum_{i=1}^n p_i - \frac{\pi}{4}=\frac{\pi}{4}. $
\end{proof}

The following property asserts the concavity of the $t$-entropy.

\begin{prop}
\label{p:3}
 Let $\bp \in [0,1]^n$ and $c>0$, then $H_c(\bp)$ is concave.
\end{prop}
\begin{proof}
We have, $\frac{\partial H_c(\bp)}{\partial p_i} = \tan^{-1}(\frac{1}{p_i^c})-\frac{c p_i^c}{1+p_i^{2c}}$. Thus, for $i \neq j$, we have, $\frac{\partial^2 H_c(\bp)}{\partial p_j \partial p_i}=0$. Also from Lemma \ref{l:2}, we have $\frac{\partial^2 H_c(\bp)}{\partial^2 p_i} \leq 0$ for all $p_i \in [0,1]$, $i=1, \dots, n$. Thus the Hessian matrix is negative definite for $\bp \in [0,1]^n$. Hence the result.
\end{proof}

The following property asserts that the $t$-entropy attains its maxima at the uniform probability vector. This property should be satisfied by any reasonable entropy since it asserts that the uncertainty of the distribution is maximum if all the outcomes are equally likely to occur. Property \ref{p:4} proves axiom (2) of Section \ref{axiom}. 
\begin{prop}
\label{p:4}
(Maximum) For any $c>0$, the entropy $H_c(\bp)$ is maximized at the uniform probability vector $(\frac{1}{n},\dots,\frac{1}{n})$.
\end{prop}
\begin{proof}
We are to maximize $H_c(\bp)$ subject to the constrain,
\begin{equation}\label{l1}
 \sum_{i=1}^n p_i=1   
\end{equation}
 and
 \begin{equation}\label{l2}
    p_i \geq 0 \text{ for all } i \in \{1,\dots,n\}. 
 \end{equation}
 It is enough to maximize $H_c(\bp)$ w.r.t. Eqn. \eqref{l1} and show that it satisfies Eqn. \eqref{l2}. The Lagrangian is given by,
 \begin{equation}\label{e3}
     L=\sum_{i=1}^n p_i \tan^{-1}\bigg(\frac{1}{p_i}\bigg) - \frac{\pi}{4}- \lambda(\sum_{i=1}^n p_i-1)
 \end{equation}
 Thus, 
 \begin{equation}\label{e4}
    \frac{\partial L}{\partial p_i}=\tan^{-1}\bigg(\frac{1}{p_i^c}\bigg)-\frac{c p_i^c}{1+p_i^{2c}}-\lambda 
 \end{equation}
 Equating the RHS of Eqn. \eqref{e4} to $0$ for all $i \in \{1,\dots,n\}$, we get,
 \begin{align*}
    & \tan^{-1}\bigg(\frac{1}{p_i^c}\bigg)-\frac{c p_i^c}{1+p_i^{2c}}=\lambda \text{ } \forall i \in \{1,\dots,n\}\\ 
    \iff & \tan^{-1}\bigg(\frac{1}{p_i^c}\bigg)-\frac{c p_i^c}{1+p_i^{2c}}=\tan^{-1}\bigg(\frac{1}{p_j^c}\bigg)-\frac{c p_i^c}{1+p_j^{2c}}\\ 
  &  \forall i,j \in \{1,\dots,n\}.
  \end{align*}
Since the function $f(x)=\tan^{-1}\bigg(\frac{1}{x^c}\bigg)-\frac{c x^c}{1+x^{2c}}$ is one-one on $[0,1]$, we have
    \begin{equation}\label{ee}
     p_i=p_j \forall i,j \in \{1,\dots,n\}     
    \end{equation}
    
 From, Eqn. \eqref{ee} and \eqref{l1}, we get, $p_i=\frac{1}{n} \forall i \in \{1,\dots,n\} $. Clearly this solution satisfies Eqn. \eqref{l2}. Hence, For any $c>0$, the entropy $H_c(\bp)$ is maximized at $\bp=(\frac{1}{n},\dots,\frac{1}{n})$.
\end{proof}

Property \ref{p:5} says that if all the outcomes are equally likely, the $t$-entropy increases with the number of outcomes. This property also should be satisfied by any reasonable entropy since the uncertainty increases with the number of outcomes if the outcomes are equally likely to occur.

\begin{prop}
\label{p:5}
If $p_1=p_2=\dots=p_n=\frac{1}{n}$, then $H_c(p)$ is an increasing function of $n$.
\end{prop}
\begin{proof}
We consider, $\bp=(p_1,\dots,p_n)$ with, $p_1=p_2=\dots=p_n=\frac{1}{n}$ and $n \geq 1$. We thus have, 
\begin{align*}
    H_c(\bp) & =\sum_{i=1}^{n} p_i \tan^{-1}\bigg(\frac{1}{p_i^c}\bigg) - \frac{\pi}{4}\\
    & = \sum_{i=1}^{n} \frac{1}{n} \tan^{-1}(n^c) - \frac{\pi}{4}\\
    & = \tan^{-1}(n^c) - \frac{\pi}{4}
\end{align*}
Since $c>0$, $H_c(\bp)$ is an increasing function of $n$.
\end{proof}

The following property tells us that an event of probability zero does not contribute to the $t$-entropy. Property \ref{pa3} proves axiom (3) of Section \ref{axiom}.

\begin{prop}
\label{pa3}
 Event of probability zero does not contribute to the entropy, i.e. for any n, $H_{c,n+1}(p_1,\dots,p_n,0)=H_{c,n}(p_1,\dots,p_n)$.
\end{prop}

\begin{proof}
\begin{align*}
H_{c,n+1}(p_1,\dots,p_n,0)  = & \sum_{i=1}^{n+1} p_i \tan^{-1}\bigg(\frac{1}{p_i^c}\bigg) - \frac{\pi}{4}\\
 = &\sum_{i=1}^n p_i \tan^{-1}\bigg(\frac{1}{p_i^c}\bigg)+0 \times \tan^{-1}(\infty)\\
& - \frac{\pi}{4}\\
 = & \sum_{i=1}^n p_i \tan^{-1}\bigg(\frac{1}{p_i^c}\bigg)+0 \times \frac{\pi}{2} - \frac{\pi}{4}\\
 = & H_{c,n}(p_1,\dots,p_n)    
\end{align*}
\end{proof}

Let $X$ and $Y$ be two random variables (each taking only finitely many values).
Property \ref{p:cond} asserts that if in addition to the information about $X$, we also have the information about $Y$, then the uncertainty of $X$ decreases. Moreover, if $X$ and $Y$ are independent, then the knowledge about $Y$ is of no help in reducing the uncertainty about $X$.  Before we proceed, let us consider the following lemma.  

\begin{lem}\label{lem3}
Convex combination of finite number of concave functions is a concave function \citep{boyd2004convex}.
\end{lem}

\begin{prop}
\label{p:cond}
 $H_c(X|Y)\leq H_c(X)$. Moreover, equality holds if $X$ and $Y$ are independent.
\end{prop}
\begin{proof}
\begin{align*}
     H_c(X|Y) & = \sum_{x\in\mathcal{X}}\sum_{y\in\mathcal{Y}} p(x,y)\tan^{-1}\bigg(\frac{1}{p(x|y)^c}\bigg) -\frac{\pi}{4} \\
    & = \sum_{x\in\mathcal{X}}\sum_{y\in\mathcal{Y}} p(y)\bigg[p(x|y)\tan^{-1}\bigg(\frac{1}{p(x|y)^c}\bigg)\bigg] -\frac{\pi}{4}\\
    & \leq \sum_{x\in\mathcal{X}}\bigg[\big(\sum_{y\in\mathcal{Y}} p(y) p(x|y)\big)\tan^{-1}\bigg(\frac{1}{\big(\sum_{y\in\mathcal{Y}} p(y) p(x|y)\big)^c}\bigg)\bigg] \\
    & \text{  }-\frac{\pi}{4}\\
    & = \sum_x p(x)\tan^{-1}\bigg(\frac{1}{p(x)^c}\bigg) -\frac{\pi}{4}\\
    & = H_c(X).
\end{align*}
The inequality follows from applying Jensen's inequality \citep{jensen1906fonctions} on the concave function  $x\tan^{-1}\big(\frac{1}{x}\big)$ ( See Lemma \ref{l:2}).\par
Note that the function $x\tan^{-1}\big(\frac{1}{x}\big)$ is strictly concave on $[0,1]$. Thus the equality in Jensen's inequality holds  $ \implies p(x|y)=p(x|y')$ for all $y,y' \in \mathcal{Y}$ $\iff$ $P(X|Y)=P(X)$ $\iff$ $X$ and $Y$ are independent.
\begin{align*}
 H_c(X|Y) & = \sum_{x\in\mathcal{X}}\sum_{y\in\mathcal{Y}} p(x,y)\tan^{-1}\bigg(\frac{1}{p(x|y)^c}\bigg) \\
 & = \sum_x\sum_y p(x)p(y|x)\tan^{-1}\bigg(\frac{1}{p(x|y)^c}\bigg)\\
 & \leq \sum_x p(x)\tan^{-1}\Bigg(\bigg(\sum_y\frac{p(x,y) \times p(y)}{p(x)\times p(x,y)}\bigg)^c\Bigg) \\
 & = \sum_x p(x)\tan^{-1}\Bigg(\bigg(\sum_y\frac{p(y)}{p(x)}\bigg)^c\Bigg) \\
 & = \sum_x p(x)\tan^{-1}\bigg(\frac{1}{p(x)^c}\bigg) = H_c(X).   
\end{align*}
Now, suppose that $X$ and $Y$ are independent. Then $p(x,y)=p(x)p(y) \forall x\in \mathcal{X}, y\in \mathcal{Y}$ and $p(x|y)=p(x)\forall x\in \mathcal{X}$.
\begin{align*}
    H_c(X|Y) & = \sum_{x\in\mathcal{X}}\sum_{y\in\mathcal{Y}} p(x,y)\tan^{-1}\bigg(\frac{1}{p(x|y)^c}\bigg)\\
    & = \sum_{x\in\mathcal{X}}\sum_{y\in\mathcal{Y}} p(x)p(y)\tan^{-1}\bigg(\frac{1}{p(x)^c}\bigg)\\
    & = \sum_{x\in\mathcal{X}}p(x)\tan^{-1}\bigg(\frac{1}{p(x)^c}\bigg)\sum_{y\in\mathcal{Y}}p(y)\\
    & = \sum_{x\in\mathcal{X}}p(x)\tan^{-1}\bigg(\frac{1}{p(x)^c}\bigg)\\
    & = H_c(X).
\end{align*}
\end{proof}

\begin{rem}
If the distribution of $X$ is replaced by the distribution of $X|Z$, then this property indicates the strong subadditivity property of the entropy. In this case, the equality holds $\iff$ $(Y,Z,X)$ is Markovian \citep{accardi1975noncommutative}.
\end{rem}

We now consider Property \ref{p:max}, which states that the joint entropy of is always greater than or equal to the marginal entropies. 
\begin{prop}
\label{p:max}
$H_c(X,Y)\geq \max\{H_c(X),H_c(Y)\}$
\end{prop}
\begin{proof}
We know that $\forall x,y$, 
\begin{align*}
         & p(x,y) \leq p(x)\\
\implies & \frac{1}{p(x,y)^c} \geq \frac{1}{p(x)^c}\\ 
\implies & \tan^{-1}\bigg(\frac{1}{p(x,y)^c}\bigg) \geq \tan^{-1}\bigg(\frac{1}{p(x)^c}\bigg).
\end{align*}
So, multiplying both sides by $p(x,y)$ and summing over all $x$ and $y$, we get,
\begin{align*}
    & \sum_x\sum_y p(x,y) \tan^{-1}\bigg(\frac{1}{p(x,y)^c}\bigg) \\
    & \geq \sum_x\sum_y p(x,y)\tan^{-1}\bigg(\frac{1}{p(x)^c}\bigg)\\
\implies & \sum_x\sum_y p(x,y) \tan^{-1}\bigg(\frac{1}{p(x,y)^c}\bigg)-\frac{\pi}{4} \\
&\geq \sum_x \tan^{-1}\bigg(\frac{1}{p(x)^c}\bigg)\sum_yp(x,y)-\frac{\pi}{4}\\
\implies & H_c(X,Y)\geq H_c(X).
\end{align*}
Similarly, we can show that $H_c(X,Y)\geq H_c(Y)$ . So, combining the two inequalities, we get, $$H_c(X,Y)\geq \max\{H_c(X),H_c(Y)\}.$$
\end{proof}

\begin{cor}
$$H_c(X_1,\dots,X_n) \geq \max\{H_c(X_1),\dots, H_c(X_n)\}.$$
\end{cor}
\begin{proof}
The Corollary easily follows from Property \ref{p:max} by using induction.
\end{proof}

The subadditivity property (Property \ref{p:cond sum}) proves axiom (4) of Section \ref{axiom}.
\begin{prop}
\label{p:cond sum}
(Subadditivity) $H_c(X,Y)\leq H_c(X)+H_c(Y|X)$.
\end{prop}
\begin{proof}
We have to prove,\\
\begingroup
\allowdisplaybreaks
\begin{align*}
 & H_c(X,Y) \leq H_c(X)+H_c(Y|X)\\
\iff & \sum_x\sum_y p(x,y)\tan^{-1}\bigg(\frac{1}{p(x,y)^c}\bigg)-\frac{\pi}{4} \leq\sum_x p(x)\tan^{-1}\bigg(\frac{1}{p(x)^c}\bigg) +\sum_x\sum_y p(x,y) \tan^{-1}\bigg(\frac{1}{p(y|x)^c}\bigg)-\frac{\pi}{2}\\
\iff & \sum_x\sum_y p(x,y) \tan^{-1}\bigg(\frac{1}{p(x,y)^c}\bigg)+\frac{\pi}{4} \leq \sum_x\sum_y p(x,y) \tan^{-1}\bigg(\frac{1}{p(x)^c}\bigg) +\sum_x\sum_y p(x,y) \tan^{-1}\bigg(\frac{1}{p(y|x)^c}\bigg)\\
\iff & \sum_x\sum_y p(x,y) \Bigg[\tan^{-1}\bigg(\frac{1}{p(x,y)^c}\bigg)+ \tan^{-1}(1)\Bigg] \leq \sum_x\sum_y p(x,y)\bigg(\tan^{-1}\bigg(\frac{1}{p(x)^c}\bigg)
 +\tan^{-1}\bigg(\frac{1}{p(y|x)^c}\bigg)\bigg)\\
\iff & \sum_x\sum_y p(x,y) \tan^{-1}\bigg(\frac{\frac{1}{p(x,y)^c}+1}{1-\frac{1}{p(x,y)^c}}\bigg) \leq \sum_x\sum_y p(x,y) \tan^{-1}\bigg(\frac{\frac{1}{p(x)^c}+\frac{1}{p(y|x)^c}}{1-\frac{1}{p(x,y)^c}}\bigg).
\end{align*}
 \endgroup
This holds if $\forall x,y$,
\begin{align*}
     & \frac{\frac{1}{p(x,y)^c}+1}{1-\frac{1}{p(x,y)^c}} \leq \frac{\frac{1}{p(x)^c}+\frac{1}{p(y|x)^c}}{1-\frac{1}{p(x,y)^c}}\\
\iff & \frac{1}{p(x,y)^c}+1 \geq \frac{1}{p(x)^c}+\frac{1}{p(y|x)^c}\\
\iff & 1+p(x,y)^c \geq \frac{p(x,y)^c}{p(x)^c}+\frac{p(x,y)^c}{p(y|x)^c}\\
\iff & 1+p(x,y)^c \geq p(y|x)^c+p(x)^c\\
\iff & 1+p(y|x)^c\times p(x)^c-p(y|x)^c-p(x)^c \geq 0\\
\iff & (1-p(x)^c)(1-p(y|x)^c)\geq 0.
\end{align*}
The last statement holds trivially since $0\leq p(x)^c,p(y|x)^c\leq 1, \forall c>0$.
\end{proof}

\begin{cor}
$H_c(X_1,\dots, X_n) \leq H_c(X_1)+H_c(X_2|X_1)+H_c(X_3|X_2,X_1)+\dots+H_c(X_n|X_1,\dots,X_{n-1})$.
\end{cor}

\begin{proof}
The Corollary easily follows from Property \ref{p:cond sum} by using induction.
\end{proof}
\begin{prop}
\label{p:sum}
$H_c(X,Y) \leq H_c(X)+H_c(Y)$.
\end{prop}

\begin{proof}
This proof follows easily from Properties \ref{p:cond} and \ref{p:cond sum}. From Property \ref{p:cond sum}, we have, $H_c(X,Y) \leq H_c(X) + H_c(Y|X) \leq H_c(X) + H_c(Y)$. The last inequality follows from Property \ref{p:cond}.
\end{proof}

\begin{cor}
$H_c(X_1,\dots, X_n) \leq \sum_{i=1}^n H_c(X_i)$.
\end{cor}
\begin{proof}
The Corollary easily follows from Property \ref{p:sum} by using induction.
\end{proof}

\begin{prop}
The conditional entropy defined in Definition \ref{d:4} satisfies the definition of conditional entropy i.e., $$H_c(X|Y)=\sum_{y\in \mathcal{Y}}p(y)H_c(X|Y=y).$$
\end{prop}
\begin{proof}
The proof follows trivially from Definition \ref{cor1}.
\end{proof}

\begin{prop}
Suppose $\bp=(p_1, \dots, p_n)$ and $\bq=(q_1,\dots,q_m)$ be two finite discrete generalized probability distribution (which are simply sequences of non-negative numbers). Let, $W(\bp)=\sum_{k=1}^n p_k$, $W(\bq)=\sum_{k=1}^m q_k$ and $\bp \cup \bq =(p_1, \dots, p_n,q_1,\dots,q_m)$. Then $H_c(\bp \cup \bq) \geq \frac{W(\bp)H_c(\bp)+W(\bq)H_c(\bq)}{W(\bp)+W(\bq)}$, provided $W(\bp)+W(\bq)\leq 1$.
\end{prop}
\begin{proof}
It is easy to note that, $H_c(\bp \cup \bq) \geq \max\{H_c(\bp),H_c(\bq)\}$. Thus,
\begin{align*}
     H_c(\bp \cup \bq)\left(W(\bp)+W(\bq)\right) = & W(\bp)H_c(\bp \cup \bq)+W(\bq)H_c(\bp \cup \bq)\\
    \geq & W(\bp)H_c(\bp)+W(\bq)H_c(\bq)\\
\end{align*}
Thus we have proved that, $H_c(\bp \cup \bq) \geq \frac{W(\bp)H_c(\bp)+W(\bq)H_c(\bq)}{W(\bp)+W(\bq)}$.

\end{proof}

R\'enyi had shown that the R\'enyi entropy can be derived from certain postulates (R\'enyi's postulates) described in \citep{renyi1961measures}. We see that the proposed $t$-entropy also satisfies the prominent postulates or their relaxed versions, most of which directly follows from its properties. Postulate $1$ is the same as property $3$, while postulate $2$ is the same as property $2$. Postulate $3$ corresponds to property $6$. Moreover, properties $12$ and $14$ are relaxed versions of postulates $4$ and $5$ of \citep{renyi1961measures}, respectively.

\section{Application to Image Segmentation}\label{image_segment}

Though entropy aims to quantify the information content, it has also found use as a measure of separation that sets apart the information into more than one connected regions \citep{al2007thresholding} in certain occasions. In particular the entropy-based image segmentation techniques have gained considerable interest within the image processing community \citep{mahmoudi2012survey,kittaneh2016average}. 
Image segmentation techniques are methods of partitioning an image into non-overlapping regions which are homogeneous with respect to some characteristics such as grayscale values or texture. There are three main groups w.r.t. image segmentation: entropic threshold, cross-entropic threshold, and fuzzy entropic threshold \citep{sezgin2004survey,csengur2006comparative}.
The general procedure adopted is to use Shannon's discrete entropy to a two-class problem, i.e., to distinguish between background and foreground, by constructing a discrete histogram. Each column in this discrete histogram represents the probability of obtaining a specified gray intensity. This method was generalized using the average entropy in place of Shannon's entropy by \cite{ferraro1999representation}.\par

 \cite{kapur1985new} proposed a method of segmenting a grayscale image into two or more segments by maximizing the posterior Shannon's entropy with respect to the threshold values. For $k$-segments image segmentation, $k$ ($>1$) probability distributions are derived from the original gray-level distribution of the image as, $\mathbf{q}_j : \frac{p_{t_{j-1}+1}}{P_j},\dots,\frac{p_{L-1}}{P_j}$,
where, $P_i=\sum_{j=t_{i-1}+1}^{t_i} p_j$, $i=1,\dots,k$ with $t_0=0$ and $t_{k}=L-1$. Let $H_c(\mathbf{q}_i)$ be the entropy of the $i$-th distribution $\mathbf{q}_i$, $i=1,\dots,k$. The posterior entropy is thus defined as $\phi(t_1,\dots, t_{k-1})=\sum_{i=1}^k H_c(\mathbf{q}_i).$

The threshold values $(t_1,\dots, t_{k-1})$ are obtained by maximizing  $\phi(t_1,\dots, t_{k-1})$ w.r.t. $(t_1,\dots, t_{k-1})$. This optimization can be carried out using grid search or other heuristic optimization techniques such as Differential Evolution \citep{storn1997differential}, Genetic Algorithms \citep{mitchell1998introduction} or Particle Swarm Optimization \citep{eberhart1995new}.
\par
The colored images are first transferred from RGB to YCbCr coordinate system \citep{gonzalez2002digital} and the Y component of the images are extracted. The Y-values can vary between $0$ to $255$. The segmentation is performed on this component of the image. This technique quite is standard in literature \citep{sarkar2015multilevel,sarkar2011differential}. 
In this paper, instead of using Shannon's entropy, we use the $t$ entropy with various values of $c$ in Kapur's method.
\par
In Kapur's method for image segmentation, the function $\phi(t_1,\dots,t_{k-1})$ is quite complicated being non-convex and non-differentiable and the explicit form is not apparent from the formulation. To overcome these difficulties, we use Differential Evolution (DE) algorithm \citep{sarkar2011differential}. DE is  a metaheuristic algorithm, 
which first generates population uniformly from the search space and then successively applies mutation, crossover and selection operation to find better candidate solution eventually leading to the optima of $\phi$.
\par
\begin{figure}[ht]
\centering
\subfloat[Original Image]{\includegraphics[width=0.3\textwidth]{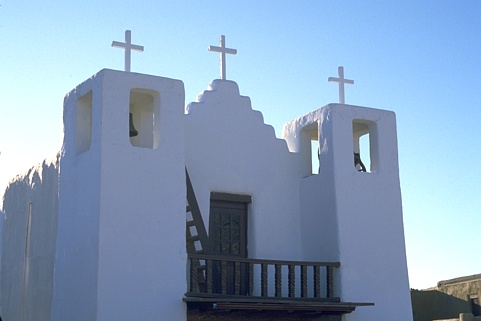}
\label{fig_first_case1}}
\hfil
\subfloat[Ground truth]{\includegraphics[width=0.3\textwidth]{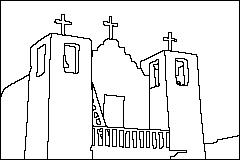}
\label{fig_second_case2}}
\hfil
\subfloat[$t$-Entropy ($c=0.1$)]{\includegraphics[width=0.3\textwidth]{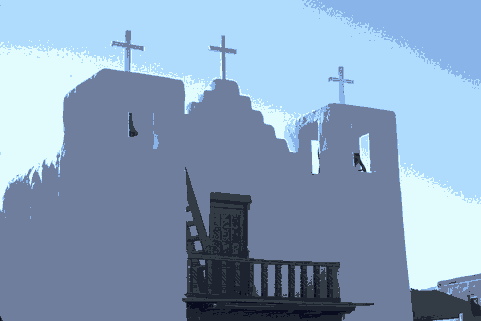}}
\label{fig_second_case3}
\hfil
\subfloat[Shannon's Entropy]{\includegraphics[width=0.3\textwidth]{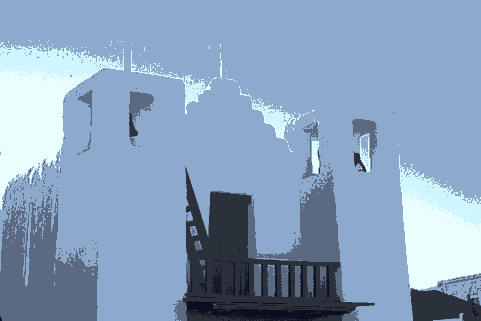}
\label{fig_second_case4}}
\hfil
\subfloat[Tsallis Entropy]{\includegraphics[width=0.3\textwidth]{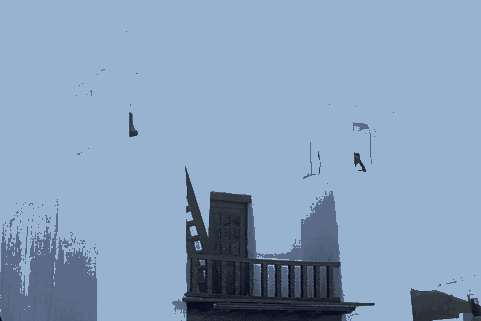}
\label{fig_second_case5}}
\hfil
\subfloat[R\'enyi Entropy]{\includegraphics[width=0.3\textwidth]{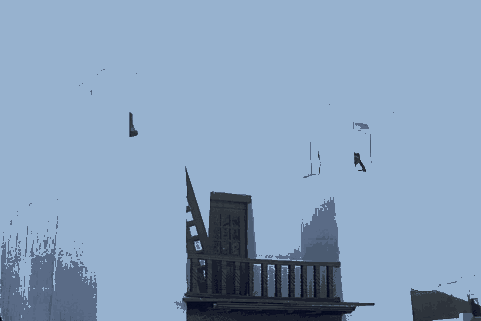}
\label{fig_second_case6}}
\caption{Segmented images obtained by Kapur's algorithm with different entropies on training image $\#24063$ with $16$ segments.}
\label{fig_sim1}
\end{figure}

To test and analyse the performance of the $t$-Entropy, we use all the 500 images from the Berkeley Segmentation Data Set and Benchmark (BSDS 500) \cite{MartinFTM01}. Each image is of size $481 \times 321$. For each image, a set
of segmented ground truth images compiled by the human observers is provided. We use the Probabilistic Rand Index (PRI), Global Consistency Error (GCE) and Variation of Information (VoI) \cite{unnikrishnan2007toward,freixenet2002yet,pantofaru2005comparison} as performance measures. All these complementary measures are considered in order to evaluate the performance of the segmentation methods. A higher value of PRI indicates better segmentation, whereas a lower value of GCE and VoI indicates the same. We run Kapur's algorithm aided with DE for $t$-Entropy (with $c=0.1$), Shannon's entropy, R\'enyi entropy (with $\alpha=2$) and Tsallis entropy (with $q=2$). The PRI, GCE and VoI between the segmentation obtained by Kapur's method and the ground truth segmentation for each image is computed. In Table \ref{tab avg bsds}, we show the average value of these indices for all the aforementioned entropies.

\begin{table}[htb]
\centering
\caption{ Comparison of average benchmark results for Kapur's Algorithm for 
different entropies (computed over 500 images from
BSDS (500) Dataset) (Best results are shown in boldface)}
\label{tab avg bsds}
\begin{tabular}{cccc}
\hline
Entropy & PRI & GCE & VoI\\
\hline
Shannon' Entropy & 0.6313 & 0.4028 & 3.0964\\
R\'enyi Entropy ($\alpha=2$) & 0.6448 & 0.4078 & 3.254\\
Tsali Entropy ($q=2$) & 0.6154 & 0.4124 & 3.1487 \\
$t$-Entropy ($c=0.1$) & \textbf{0.6527} & \textbf{0.3858} & \textbf{3.0695}\\
$t$-Entropy ($c=10$) & 0.6285 & 0.407 & 3.0939\\
$t$-Entropy ($c=50$) & 0.6345 & 0.4006 & 3.0776
\\ \hline
\end{tabular}
\end{table}

It can be easily seen from Table \ref{tab avg bsds} that Kapur's method with $t$-Entropy performs better than that with the other entropies in terms of the PRI, GCE and VoI indices. The outcomes of DE based Kapur's method for all the competiting methods on one of the BSDS (500) images are shown in Fig. \ref{fig_sim1}. It can be easily seen that Kapur's method with $t$-Entropy (with $c=0.1$) is closer to the ground truth than that with the other peer entropies.
\section{Application to Statistics: A Robust Estimator based on \textit{t}-Entropy}\label{stat}
In this section, we will show an application of the $t$-entropy to the statistical point estimation. We will first derive a relative entropy based on the $t$-entropy measure and construct an estimator based on the same. For simplicity, we set $c=1$.\par
Formally, let $\mu$ be a measure on $\mathcal{X}$, which dominates two other measures $P$ and $Q$, i.e. $P,Q\ll \mu$. By the  Radon-Nikodym Theorem \citep{billingsley2008probability}, $P$ and $Q$ possess derivatives $p$ and $q$, i.e. $p=\frac{dP}{d\mu}$ and $q=\frac{dQ}{d\mu}$. The relative entropy between the two measures $P$ and $Q$ is defined as 
\begin{equation}
    D_t(P||Q)=\int_{\mathcal{X}}p \tan^{-1}\bigg(\frac{p}{q}\bigg)d\mu-\frac{\pi}{4}.
\end{equation}
If $P \ll Q$, one can write the above equation as 
\begin{equation}
    D_t(P||Q)=\int_{\mathcal{X}} \tan^{-1}\bigg(\frac{dP}{dQ}\bigg)dP-\frac{\pi}{4}.
\end{equation}
If we take $\mu$ to be the Lebesgue measure and $p$ and $q$ as the probability density functions of $P$ and Q respectively, the relative entropy between $P$ and $Q$ boils down to $D_t(P||Q)=\int_{\mathcal{X}} \tan^{-1}\Big(\frac{p(x)}{q(x)}\Big)p(x)dx-\frac{\pi}{4}=\mathbb{E}_{X\sim P}\bigg[\tan^{-1}\Big(\frac{p(X)}{q(X)}\Big)\bigg]-\frac{\pi}{4}.$ We note that this relative entropy is an $f$-divergence \citep{csiszar1975divergence}, as it can be written as $\int_{\mathcal{X}} f\big(\frac{p(x)}{q(x)}\big)q(x)dx$, where $f(x)=x \tan^{-1}(x)-\pi/4$. We note that $f(1)=0.$ \par
We will now discuss the application of this divergence in the context of point estimation. Suppose $X_1,\dots,X_n$ be independent and identically distributed according to some distribution $p_\theta$ with $\theta \in \Theta$. Our goal is to estimate $\theta$, based on the observed data. One way to get an estimate of $\theta$ is to consider the divergence between the two distributions $\hat{p}_n$ and $p_\theta$. Here $\hat{p}_n$ is an estimate of the distribution based on $x_1,\dots,x_n$. 
We define our estimate for $\theta$ based on the data $x_1,\dots,x_n$ as follows: 
\begin{equation}
\label{estimate}
\hat{\theta}_t=\argmin_{\theta \in \Theta}D_t(\hat{p}_n||p_\theta).
\end{equation}
One can take $\hat{p}_n=\frac{1}{n}\delta(x_i)$, $\delta(\cdot)$ being the Dirac delta function, which denotes the empirical distribution of $x_1,\dots,x_n$. One can also take the kernel density estimator 
based on the data $x_1,\dots,x_n$. 
\subsection{Existence and Consistency of the \textit{t}-Estimator}
Let $\mathcal{G}$ denote the set of all distributions having density w.r.t. some dominating measure $\mu$.
Let $\mathcal{F}=\{F_\theta|\theta \in \Theta\} \subset \mathcal{G}$ be a family of distributions characterized by the parameter $\theta$. Let $\frac{d F_\theta}{d \mu}=f_\theta$. For any distribution $G$ having density $g$ w.r.t. $\mu$, the functional $T(G)$ is defined by the requirement $D_t(G||F_{T(G)})=\inf_{\theta \in \Theta} D_t(G||F_\theta).$
The following theorem asserts the existence and consistency properties of the $t$-estimator.
\begin{thm}
Let the parametric family $\mathcal{F}$ be identifiable and let $\Theta$ be a compact subset of $\mathbb{R}^p$. Let $f_\theta(x)$ be continuous a.s. $[\mu]$. Then,
\begin{enumerate}
    \item for all $G \in \mathcal{G}$, $T(G)$ exists.
    \item if $T(G)$ is unique, then the functional $T$ is continuous at $G$ under the total variation topology \Big(i.e., $T(G_n) \to T(G)$, whenever $\int|g_n-g|d\mu \to 0$. Here $g_n$ is the density of $G_n$\Big). 
    \item $T(F_\theta)=\theta$ for all $\theta \in \Theta$.
\end{enumerate}
\end{thm}
\begin{proof}
\textbf{Proof of part (1):}
     Let $t_n \to t$ be a sequence of parameter values in $\Theta$. Then, 
    \begin{align*}
        |D_t(G,F_{t_n})-D_t(G,F_{t})| & =\bigg|\int \bigg[\tan^{-1}\bigg(\frac{g(x)}{f_{t_n}(x)}\bigg)-\tan^{-1}\bigg(\frac{g(x)}{f_{t}(x)}\bigg)\bigg]g(x)d\mu(x)\bigg|\\
        & \leq \int \bigg|\tan^{-1}\bigg(\frac{g(x)}{f_{t_n}(x)}\bigg)-\tan^{-1}\bigg(\frac{g(x)}{f_{t}(x)}\bigg)\bigg|g(x)d\mu(x).
    \end{align*}
    The last term goes to $0$ by a simple application of the Dominated Convergence Theorem (DCT) \citep{billingsley2008probability}. Thus the function $h(t)=D_t(G||F_t)$ is continuous on $\Theta$, which is a compact set. Thus $h(\cdot)$ attains its minimum on $\Theta$.\par
    \textbf{Proof of part (2):}
     Let $\{G_n\}$ converges to $G$ in the total variation sense, i.e. $\int|g_n(x)-g(x)|d\mu(x) \to 0$ as $n \to 0$ Let $h_n(t)=D_t(G_n||F_t)$, $\theta_n=T(G_n)$ and $\theta=T(G)$. We will first show that $h_n(t) \to h(t)$.
     \begingroup
     \allowdisplaybreaks
    \begin{align*}
        |h_n(t)-h(t)|   & =\bigg|\int \bigg[\tan^{-1}\bigg(\frac{g_n(x)}{f_{t}(x)}\bigg)g_n(x)  -\tan^{-1}\bigg(\frac{g(x)}{f_{t}(x)}\bigg)g(x)\bigg]d\mu(x)\bigg|\\
        & = \bigg|\int \bigg[\tan^{-1}(\xi_x)+\frac{\xi_x}{1+\xi_x^2}\bigg](g_n(x)-g(x))d\mu(x)\bigg|.
    \end{align*}
    \endgroup
    Here $\xi_x$ lies between $g(x)$ and $g(x)$. The last equality follows form applying Taylor's on the function $x\tan^{-1}(x)$. We also note that $\bigg|\tan^{-1}(\xi)+\frac{\xi}{1+\xi^2}\bigg|$ is bounded above by 2. Thus,
    $|h_n(t)-h(t)|\leq 2\int|g_n(x)-g(x)|d\mu(x)$ for all $t \in \Theta$. Thus, $\lim_{n \to \infty}\sup_{t \in \Theta}|h_n(t)-h(t)|=0.$ Now note that $ h(\theta)-h_n(\theta_n)\leq h(\theta_n)-h_n(\theta_n)$ and $h_n(\theta)-h(\theta)\leq h_n(\theta)-h(\theta).$ It is thus easy to conclude that
    \begin{align}
    |h_n(\theta_n)-h(\theta)| \nonumber & \leq |h_n(\theta_n)-h(\theta_n)|+|h_n(\theta)-h(\theta)| \nonumber\\
    & \leq 2 \sup_{t \in \Theta} |h_n(t)-h(t)|. \label{m3}
    \end{align}
    Thus, $\lim_{n \to \infty} h_n(\theta_n)=h(\theta)$.\par 
    It remains to be shown that $\theta_n \to \theta$. Assume the contrary. Then, appealing to the compactness of $\Theta$, there exists a subsequence, say $\{\theta_{n_l}\}_{l=1}^\infty$ of $\{\theta_{n}\}_{n=1}^\infty$, such that $\theta_{n_l} \to \theta_1$, where, $\theta_1 \neq \theta$. Since $h$ is continuous, $h(\theta_{n_l}) \to h(\theta_1)$. From \eqref{m3}, we get, $h(\theta_1)=h(\theta)$. This gives us a contradiction, since $T(G)$ is assumed to be unique. Thus, $\theta_n \to \theta$.\par
    \textbf{Proof of part (3):}
        Since, the parametric family $\{F_\theta: \theta \in \Theta\}$ is identifiable, $D_t(F_\theta||F_{t}$) attains the value zero at $t=\theta$, uniquely. Thus $T(F_\theta)=\theta$, uniquely. 
\end{proof}
\subsection{Robustness of the \textit{t}-Estimator}
To theoretically assert the robustness of an estimator, we will use the concept of breakdown point  \citep{hampel1971general,donoho1983notion}. The breakdown point of a functional can be thought of as the smallest proportion of contamination in the data that can cause an arbitrary extreme value in the estimate. To investigate the robustness of the $t$-Estimator, we consider the contaminated sequence of distributions, $H_{\epsilon,n}=(1-\epsilon)G+\epsilon K_n.$
Here $\{K_n\}_{n=1}^\infty$ is some sequence of contaminating distributions and $\epsilon$ is the contaminating proportion. Let the density of $K_n$ w.r.t the Lebesgue measure be $k_n$. Following the notion of Simpson \citep{simpson1987minimum}, we say that a breakdown occurs in $T$ at $\epsilon$ level contamination, if there exists a sequence $K_n$, for which, $|T(H_{\epsilon,n})-T(G)|\to \infty$ as $n \to \infty$.\par
Let $\theta_n=T(H_{\epsilon,n})$. Following the works of Park and Basu \citep{park2004minimum}, we make the following standard assumptions for our breakdown point analysis.
\begin{itemize}
    \item[A1.] $\int \min\{g(x),k_n(x)\} \to 0 $ as $n \to \infty$. 
    \item[A2.] $ \int \min\{f_\theta(x),k_n(x)\} \to 0 $ as $n \to \infty$, uniformly for $|\theta| \leq c$, where $c$ is some fixed positive constant.
    \item[A3.] $ \int \min\{f_{\theta_n}(x),k_n(x)\} \to 0 $ as $n \to \infty$, if $|\theta_n| \to \infty$.
\end{itemize}

\begin{thm}
Under assumptions \textit{A1-A3}, the asymptotic breakdown point $\epsilon^*$ of the $t$-functional is at least $\frac{1}{2}$ at the model.
\end{thm}
\begin{proof}
Let there be a sequence $K_n$, for which $|\theta_n| \to \infty$. Let $A_n=\big\{x: g(x)> \max\{k_n(x),f_{\theta_n}(x)\}\big\}$. Thus we have,
\begingroup
\allowdisplaybreaks
\begin{align*}
 D_t(H_{\epsilon,n}||F_{\theta_n}) =&\int_{A_n}\tan^{-1}\bigg(\frac{(1-\epsilon)g(x)+\epsilon k_n(x)}{f_{\theta_n}(x)}\bigg)((1-\epsilon)g(x)+\epsilon k_n(x))dx \\
& + \int_{A_n^C}\tan^{-1}\bigg(\frac{(1-\epsilon)g(x)+\epsilon k_n(x)}{f_{\theta_n}(x)}\bigg)((1-\epsilon)g(x)+\epsilon k_n(x))dx -\pi/4.
\end{align*}
\endgroup
From assumption \textit{A1}, we get, $\int_{A_n}k_n(x)\to 0$ and from A3, we have, $\int_{A_n}f_{\theta_n}(x) \to 0$. For notational simplicity, we define, $C(g(x),f(x))=\tan^{-1}(g(x)/f(x))g(x)$. We note that under the probability measures induced by the densities $k_n(\cdot)$ and $f_{\theta_n}(\cdot)$, $A_n$ converges to a set with zero probability. Thus on $A_n$, $C(h_{\epsilon,n}(x),f_{\theta_n}(x)) \to (1-\epsilon)\frac{\pi}{2}g(x).$ Thus by dominated convergence theorem, 
\begin{equation}
    \bigg|\int_{A_n}C(h_{\epsilon,n}(x),f_{\theta_n}(x))dx - \int_{g>0} (1-\epsilon)\frac{\pi}{2}g(x)dx\bigg| \to 0.
\end{equation}
Thus, $ \big|\int_{A_n}C(h_{\epsilon,n}(x),f_{\theta_n}(x))dx - (1-\epsilon)\frac{\pi}{2}dx\big| \to 0.$
Again from A1 and A3, $\int_{A_n^C}g(x) \to 0$. Again by DCT, $\int_{A_n^C}C(h_{\epsilon,n}(x),f_{\theta_n}(x)) \to \int \tan^{-1}\bigg(\frac{\epsilon k_n(x)}{f_{\theta_n}(x)}\bigg) \epsilon k_n(x)dx \geq \epsilon \tan^{-1} (\epsilon).$
The last inequality follows from applying Jensen's inequality on the function $\tan^{-1}(1/x)$. Thus we have,
\begin{equation}
    \liminf_{n \to \infty} D_t(H_{\epsilon,n}||F_{\theta_n}) \geq (1-\epsilon) \frac{\pi}{2}+\epsilon \tan^{-1} (\epsilon)-\pi/4.
\end{equation}
Let, $a_1(\epsilon)=(1-\epsilon) \frac{\pi}{2}+\epsilon \tan^{-1} (\epsilon)-\pi/4$.
Now let $\theta^*$ be the minimizer of $\int C((1-\epsilon)g(x),f_\theta(x))$. For any fixed $\theta \in \Theta$, we define, $B_n=\{x: k_n(x) > \max(g(x),f_\theta(x))\}$. From \textit{A1}, we get, $\int_{B_n} g(x) \to 0$ and from \textit{A2}, we get, $\int_{B_n} f_\theta(x) \to 0$. Similarly, from \textit{A1} and \textit{A2}, we have $\int_{B_n^C} k_n(x) \to 0$. Thus, under $k_n(\cdot)$, $B_n$ converges to a set with zero probability. Hence, by applying DCT, we get,
\begin{equation}
    \int_{B_n} C(h_{\epsilon,n}(x),f_\theta(x))dx \to \int_{\{x:k_n(x) > 0\}} \frac{\pi}{2} \epsilon k_n(x)dx=\frac{\epsilon \pi}{2}.
\end{equation}
Similarly, {\small $\int_{B_n^C} C(h_{\epsilon,n}(x),f_\theta(x))dx \to \int C((1-\epsilon)g(x),f_\theta(x))dx.   $}
 Hence, we have, 
\begin{equation}
\label{eq23}
    \lim_{n \to \infty} D_t(H_{\epsilon,n}||F_{\theta}) \geq \frac{\epsilon \pi}{2}+\inf_{\theta \in \Theta} \int C((1-\epsilon)g(x),f_\theta(x))dx - \pi/4.
\end{equation}

The equality on \eqref{eq23} holds if $\theta=\theta^*$. If $g(\cdot)=f_{\theta_t}(\cdot)$,
\begin{equation}
    \int C((1-\epsilon)f_{\theta_t}(x),f_\theta(x)) \geq (1-\epsilon) \tan^{-1}(1-\epsilon).
\end{equation}
The equality holds if $\theta^*=\theta_t$ and in that case, $\lim_{n \to \infty} D_t(H_{\epsilon,n}||F_\theta)=\frac{\epsilon \pi}{2}+ (1-\epsilon) \tan^{-1}(1-\epsilon)-\pi/4=a_3(\epsilon).$
Hence, asymptotically, there is no breakdown for a $\epsilon$ level contamination, if $a_2(\epsilon) < a_1(\epsilon)$, which occurs when $\epsilon<\frac{1}{2}.$ 
\end{proof}

\section{Application to Clustering}\label{clustering}
Clustering refers to the task of partitioning a collection of datapoints into some homogeneous groups \citep{wong2015short,xu2005survey}. $k$-means \citep{macqueen1967some} is by far the most popular algorithm for data clustering. Consider a dataset $\mathcal{X}=\{\bx_1,\dots,\bx_n\} \subset \mathbf{R}^p$. In order to partition the dataset into $k$ disjoint groups, $k$-means formulates the problem as the minimization of the following function:
\begin{equation}
\label{kmeans}
    f_{k-means}(\Theta)=\sum_{i=1}^n \min_{1 \leq j \leq k} \|\bx_i-\btheta_j\|_2^2,
\end{equation}
where $\Theta=\{\btheta_1,\dots,\btheta_k\} \subset \mathbb{R}^p$ is the set of all the $k$ centroids. This objective function  can be interpreted as the within cluster sum of squares. Lloyd's algorithm \citep{lloyd1982least} is a popular coordinate descent algorithm to optimize \eqref{kmeans}.   Despite its wide-spread application, $k$-means is notoriously unsuitable for high-dimensional datasets, where only a handful of features are relevant in revealing the cluster structure of the dataset \citep{chakraborty2020detecting}. To tackle this problem, researchers have often resorted to the concept of feature weighting \citep{de2016survey}. 
\par
\citep{huang2005automated} proposed the feature weighted $k$-means ($W$-$k$-means) clustering method, which formulates the clustering problem as the minimization of the following objective function:
\begin{equation}
\label{wkmeans}
f_{W-k-means}(\Theta,\bw)=\sum_{i=1}^n \min_{1 \leq j \leq k} \sum_{l=1}^p w_l^\beta (x_{il}-\theta_{jl})^2,
\end{equation}
where $\bw=(w_1,\dots,w_p)'$ denotes the vector of feature weights. Objective function \eqref{wkmeans} is minimised w.r.t. the constraint that $\sum_{l=1}^{p}w_l=1$ and $w_l \geq 0$ for all $l=1,\dots,p$. \cite{jing2007entropy} further extended this idea to incorporate cluster specific feature weighting along with an entropy regularization on the feature weights. This technique, referred to as Entropy Weighted $k$-means ($EW$-$k$-means), is particularly useful if the clusters lie in different subspaces of $\mathbb{R}^p$. The formulation by \cite{jing2007entropy} of the clustering objective function is as follows:
\begin{align}
f_{EW-k-means}(\Theta,W) & =\sum_{i=1}^n \min_{1 \leq j \leq k} \sum_{l=1}^p W_{jl}  (x_{il}-\theta_{jl})^2-\lambda \sum_{j=1}^kH_{Shannon}(W_{j,\cdot}), \label{EWK}
\end{align}
where $W \in \mathbb{R}^{k \times p}$ denotes the matrix, whose $j$-th row, $W_{j,\cdot}$ contains the feature weights for the $j$-th cluster. The objective function \eqref{EWK} is minimized w.r.t. the following constraints,
\begin{align}
&W_{jl} \geq 0 \text{ for all } j=1,\dots,k \text{ and }  l=1,\dots,p \label{c1}\\
&\sum_{l=1}^pW_{jl}=1 \text{ for all } j=1,\dots,k. \label{c2}
\end{align}
In our formulation, we replace Shannon's entropy with $t$-entropy in \eqref{EWK}. For the sake of simplicity, we take $c=1$. Our clustering objective is thus given by, 
\begin{align}
f(\Theta,W) & =\sum_{i=1}^n \min_{1 \leq j \leq k} \sum_{l=1}^p W_{jl} (x_{il}-\theta_{jl})^2-\lambda \sum_{j=1}^k\sum_{l=1}^{p} W_{jl}\tan^{-1}\bigg(\frac{1}{W_{jl}}\bigg). \label{tclus}
\end{align}
Objective function \eqref{tclus} is minimised subject to the constraints \eqref{c1} and \eqref{c2}.\par
\subsection{Real Data Analysis}
 We consider nine benchmark datasets from the UCI machine learning repository \citep{Dua:2019}, Keel repository \citep{alcala2011keel} and Arizona State University  feature selection repository \citep{li2018feature} to validate the performance of our formulation. A brief description of the datasets along with their sources is provided in Table \ref{tab:source}. In particular the datasets GLIOMA and LIBRAS are quite challenging as $p \gg n$ and $p \approx n $ on these two datasets respectively. As a cluster validation index, we use the Normalized Mutual Information (NMI) \citep{vinh2010information} and the Adjusted Rand Index (ARI) \citep{hubert1985comparing} between the ground truth and the partitioning obtained by the algorithm. A value of 1 indicates complete match and a value of 0 indicates complete mismatch.  To compare our method, we choose the $k$-means, $W$-$k$-means, $EW$-$k$-means with Shannon's entropy and  Minkowski Weighted $k$-means \citep{de2012minkowski}. We run each algorithm 20 times on each of the datasets until convergence and report the average NMI and ARI values in Tables \ref{nmi} and \ref{ARI} respectively. The best performing algorithm for each of the datasets are bold-faced. It can be observed that in terms of both the indices, $EW$-$k$-means with $t$-entropy provides a better clustering than the peer algorithms in most of the benchmark datasets.\par
 \begin{table}
    \centering
    \caption{Source and Description of the Datasets}
    \label{tab:source}
    \begin{tabular}{|l|c|c|c|c|}
    \hline
        Dataset & Source & $n$ & $p$ & $k$   \\
        \hline
         Iris & \href{https://archive.ics.uci.edu/ml/datasets.php}{UCI Repository} & 150 & 4 & 3\\
         WDBC & \href{https://sci2s.ugr.es/keel/category.php?cat=clas}{Keel Repository} & 569 & 30 & 2\\
         Mammographic & \href{https://sci2s.ugr.es/keel/category.php?cat=clas}{Keel Repository} & 830 & 5 & 2\\
         Newthyroid & \href{https://sci2s.ugr.es/keel/category.php?cat=clas}{Keel Repository} & 215 & 5 & 3\\
         Heart & \href{https://sci2s.ugr.es/keel/category.php?cat=clas}{Keel Repository} & 270 & 13 & 2\\
         Hepatitis & \href{https://sci2s.ugr.es/keel/category.php?cat=clas}{Keel Repository} & 80 & 19 & 2\\
         Mice Protein & \href{https://archive.ics.uci.edu/ml/datasets.php}{UCI Repository} & 1080 & 77 & 8\\
         GLIOMA & \href{http://featureselection.asu.edu/index.php}{ASU Repository} & 50 & 4434 & 4\\
         LIBRAS & \href{https://archive.ics.uci.edu/ml/datasets.php}{UCI Repository} & 144 & 90 & 6 \\
         \hline
    \end{tabular}
\end{table}

\begin{table}
    \centering
        \caption{Comparison of NMI Values on Real-Life Datasets (Best results are shown in boldface)}
    \label{nmi}
    \begin{tabular}{|l|c|c|c|c|c|}
    \hline
      Datasets   &  $k$-means & $W$-$k$-means & $MW$-$k$-means & $EW$-$k$-means(Shannon) & $EW$-$k$-means($t$)\\
      \hline
       Iris  & 0.7244(5) & \textbf{0.7885}(1) & 0.7513(4) & 0.7741(2) & 0.7582(3)\\
       WDBC & 0.4636(2.5) & 0.0056(4) & 0.0016(5) & 0.4636(2.5) & \textbf{0.5687}(1)\\
       Mammographic & 0.1074(3) & 0.0194(4) & 0.0074(5) & 0.2339(2) & \textbf{0.2577}(1)\\
       Newthyroid & 0.4031(3) & 0.2625(4) & 0.1516(5) & 0.5072(2) & \textbf{0.6872}(1)\\
       Heart & 0.0174(4.5) & 0.1096(2) & 0.0370(3) & 0.0174(4.5) & \textbf{0.3101}(1)\\
       Hepatitis & 0.0005(5) & 0.1493(2) & 0.0578(3) & 0.0155(4) & \textbf{0.2603}(1)\\
       Mice Protein & 0.2508(3) & 0.2029(4) & 0.0759(5) & 0.2515(2) & \textbf{0.2800}(1)\\
       GLIOMA & 0.4468(2) & 0.4274(3) &  0.3977(4) & 0.2652(5) & \textbf{0.4892}(1)\\
       LIBRAS & 0.5387(5) & 0.5765(3) & 0.5473(4) & 0.5979(2) & \textbf{0.6554} (1)\\
       \hline
       Average Rank & 3.67 & 3 &4.22 & 2.89 & \textbf{1.22} \\
       \hline
    \end{tabular}
\end{table}

\begin{table}
    \centering
        \caption{Comparison of ARI Values on on Real-Life Datasets (Best results are shown in boldface)}
    \label{ARI}
    \begin{tabular}{|l|c|c|c|c|c|}
    \hline
    Datasets   &  $k$-means & $W$-$k$-means & $MW$-$k$-means & $EW$-$k$-means(Shannon) & $EW$-$k$-means($t$)\\
      \hline
       Iris & 0.6707(5) & 0.7484(2) & 0.7027(4) & \textbf{0.7427}(1) & 0.7302(3)\\
       WDBC & 0.4904(2.5) & 0.0127(4) & 0.0005(5) & 0.4904(2.5) &  \textbf{0.6850}(1)\\
       Mammographic & 0.1367(3) & 0.0005(4.5) & 0.0005(4.5) & 0.2433(2) & \textbf{0.3093}(1) \\
       Newthyroid &  0.4827(3) & 0.1641(5) & 0.2554(4) & 0.5502(2) & \textbf{0.7656}(1)\\
       Heart & 0.0264(4) & 0.1323(2) & 0.0445(3) & 0.0262(5) & \textbf{0.4018}(1)\\
       Hepatitis & 0.0168(5) & 0.2711(3) & 0.1099(4) & \textbf{0.4871}(1) & 0.4102(2)\\
       Mice Protein & 0.1390(3) & 0.1033(5) & 0.1193(4) & 0.1483(2) & \textbf{0.1529}(1)\\
       GLIOMA & 0.2806(3) & 0.2881(2) & 0.2488(4) & 0.1078(5) & \textbf{0.3707}(1)\\
       LIBRAS & 0.3588(4) & 0.3789(3) & 0.3458(5) & 0.4729(2) & \textbf{0.5346}(1) \\
       \hline
       Average Rank & 3.61 & 3.39 & 4.17 & 2.28 & \textbf{1.33}\\
       \hline
    \end{tabular}
\end{table}
 \subsection{Case Study on Libras Data}
 We evaluate the performance of various clustering algorithms on the LIBRAS movement dataset. The dataset is collected from the UCI machine learning repository \citep{Dua:2019}. The dataset consists of 15 classes, each class referring to a type of hand movement. Each class contains 24 observations and each observation has 90 features consisting of the coordinates of hand movements. Since there are overlaps between the clusters, as observed by \cite{wang2018sparse}, we consider six clusters: vertical swing (labeled as 3), anti-clockwise arc (labeled as 4), clockwise arc (labeled as 5), horizontal straightline (labeled as 7), horizontal wavy (labeled as 11), and vertical wavy  (labeled as 12) in the original dataset.
 For better visualization, in Fig. \ref{fig_cluster}, we show the $t$-SNE plots \citep{maaten2008visualizing} of the LIBRAS dataset, color-coded with the partition obtained by each of the peer algorithms. It is clear from Fig. \ref{fig_cluster} that the $EW$-$k$-means with $t$-entropy resembles the ground truth, while other peer algorithms fail to do so.
 \begin{figure}[ht]
\centering
\subfloat[Ground Truth]{\includegraphics[width=0.3\textwidth]{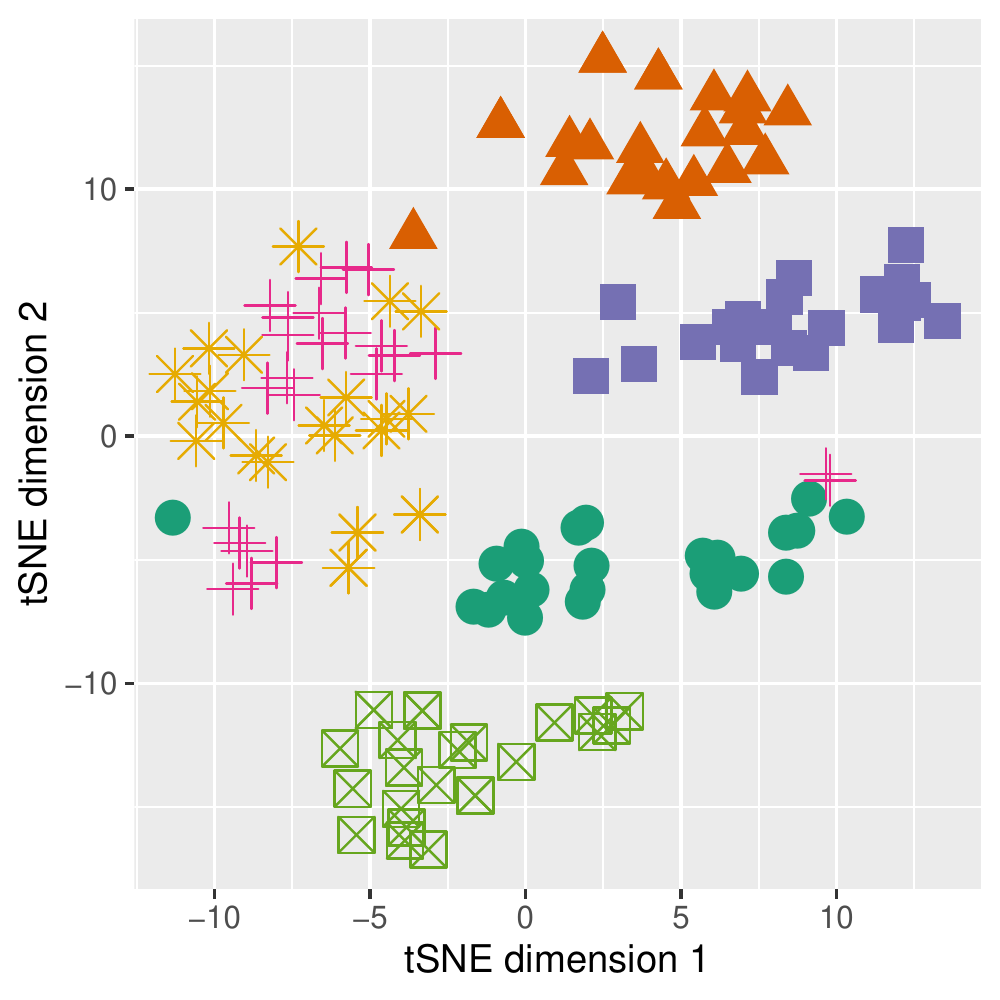}}
\hfil
\subfloat[$k$-means]{\includegraphics[width=0.3\textwidth]{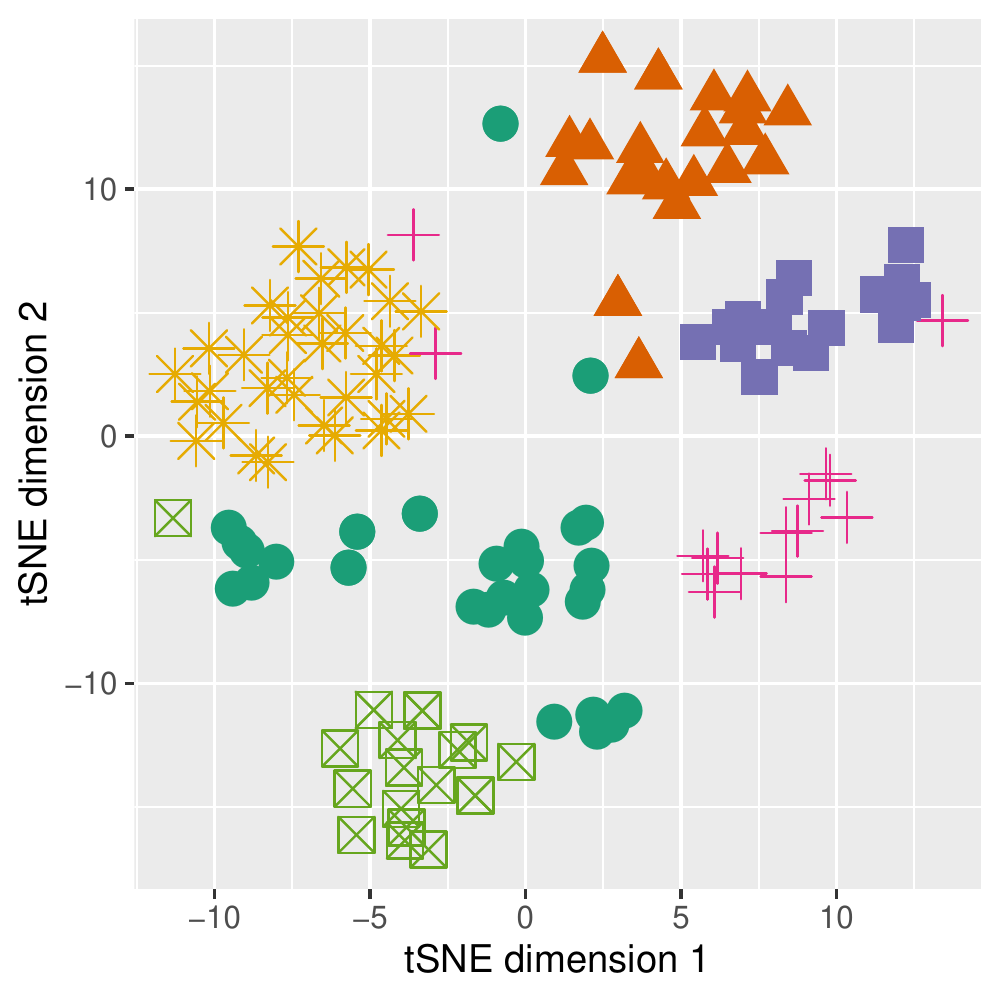}}
\hfil
\subfloat[$W$-$k$-means]{\includegraphics[width=0.3\textwidth]{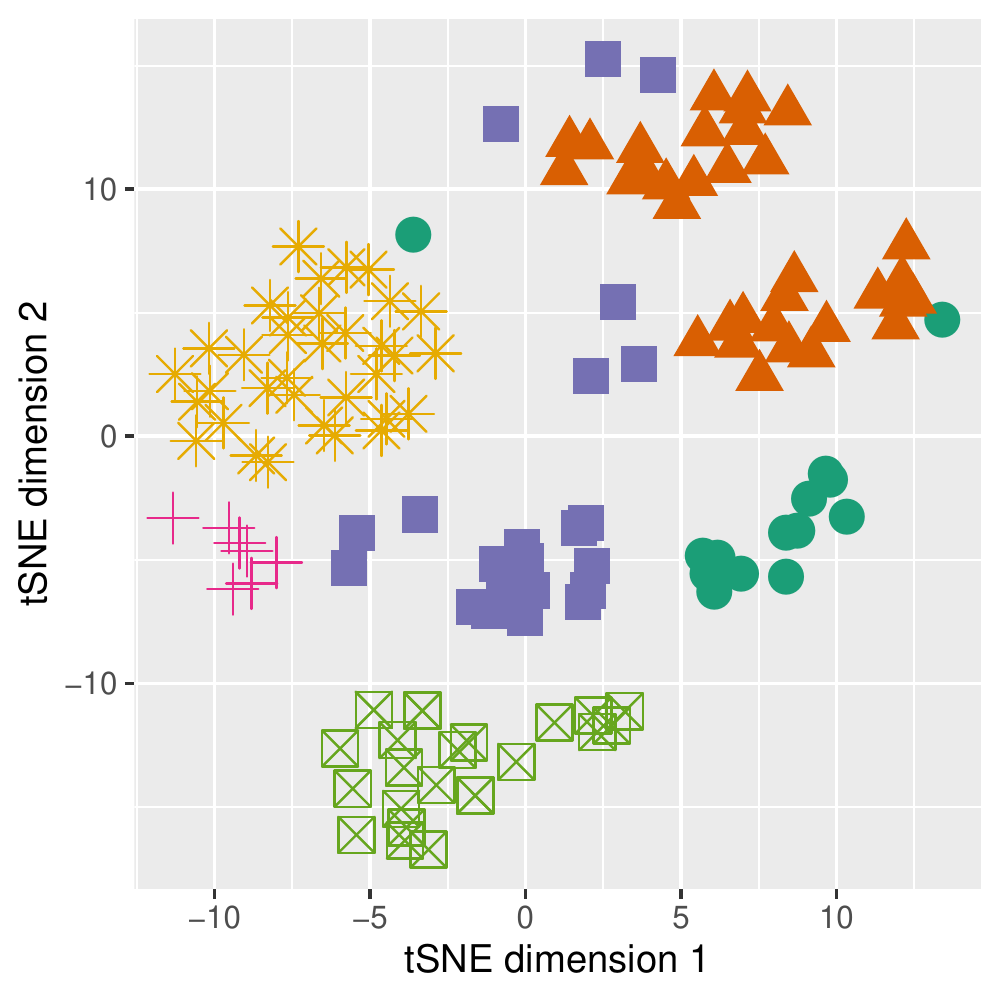}}
\hfil
\subfloat[$MW$-$k$-means]{\includegraphics[width=0.3\textwidth]{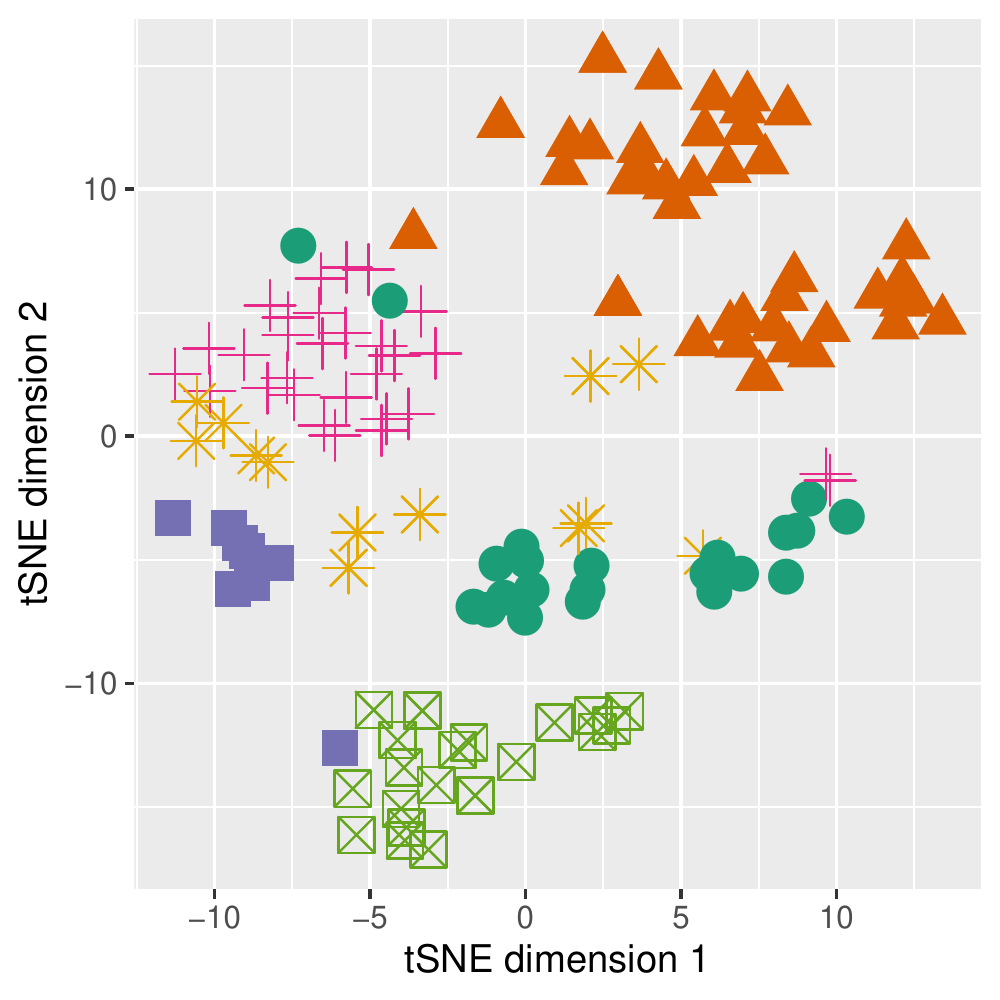}}
\hfil
\subfloat[$EW$-$k$-means (with Shannon's Entropy)]{\includegraphics[width=0.3\textwidth]{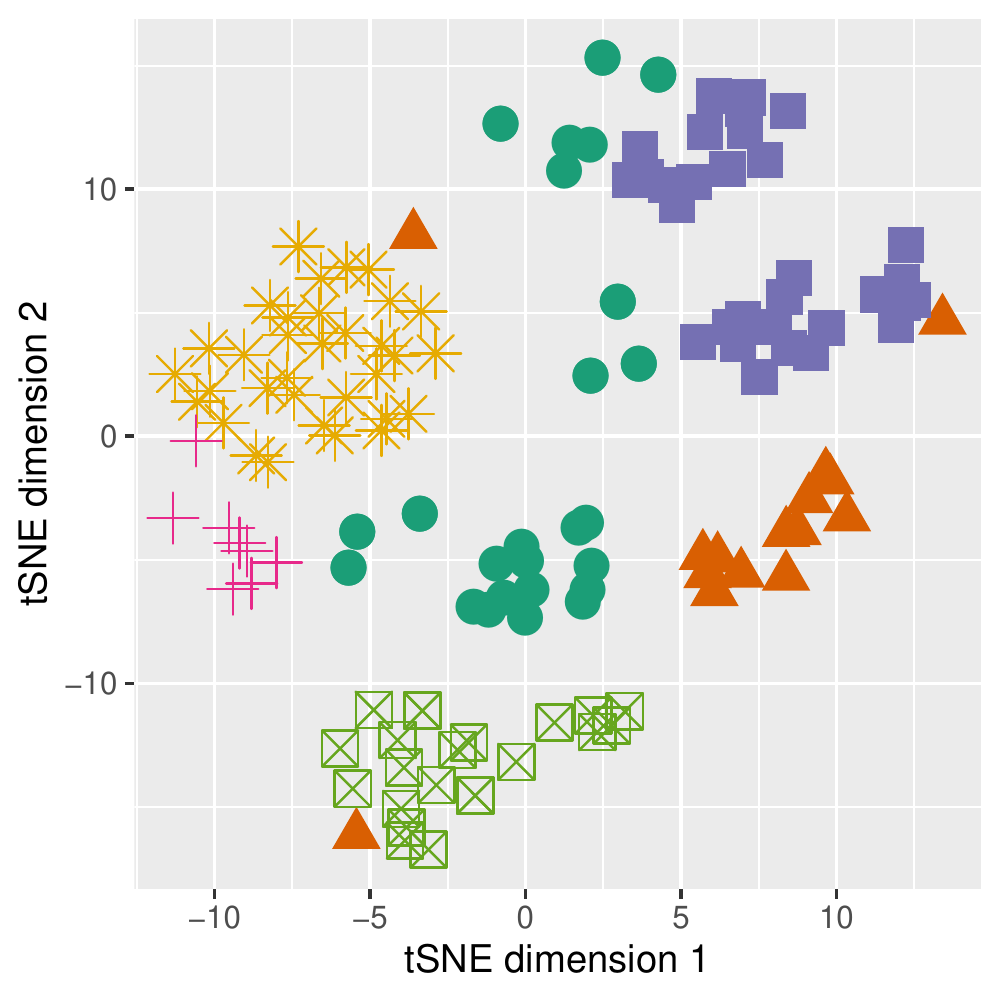}}
\hfil
\subfloat[$EW$-$k$-means (with $t$-Entropy)]{\includegraphics[width=0.3\textwidth]{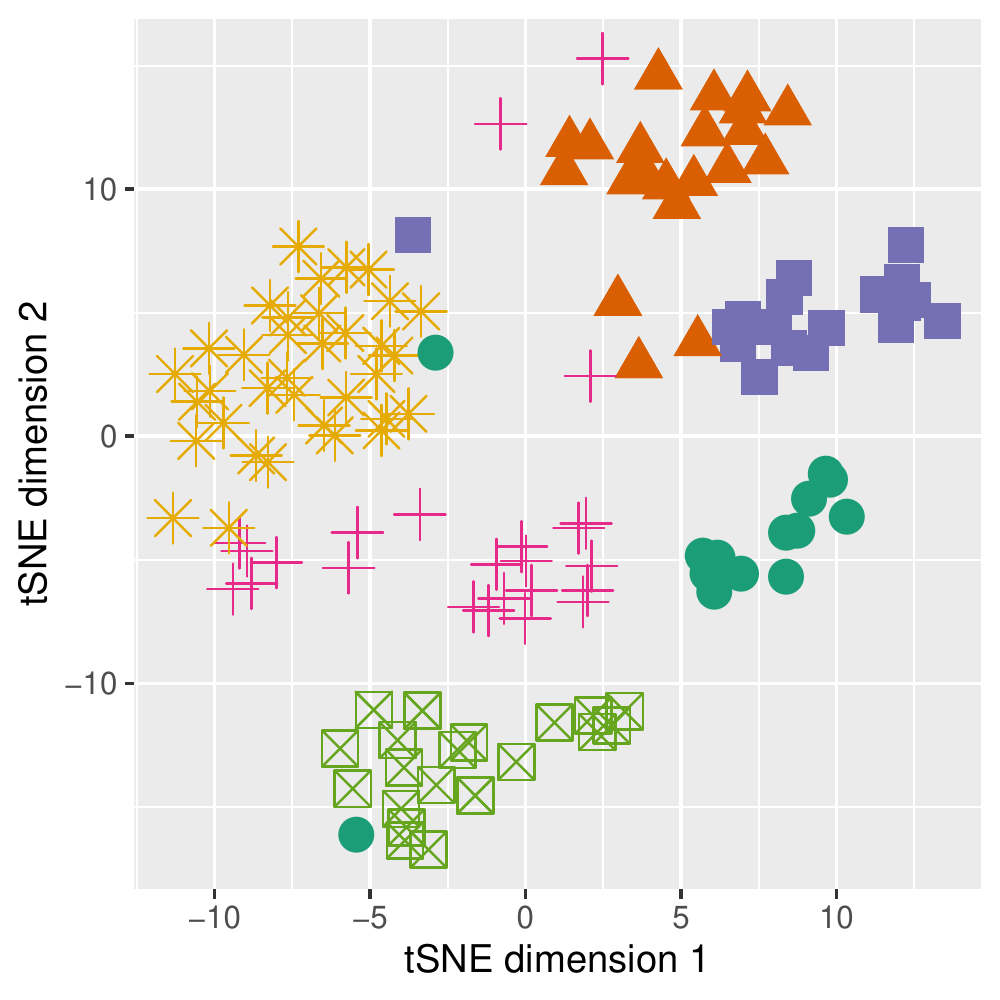}}
\hfil
\centering
\caption{$t$-SNE plots for LIBRAS movement dataset, showing the performance of various clustering algorithms.}
\label{fig_cluster}
\end{figure}
\section{Conclusion}\label{discussion}
We propose a new class of entropy measures, called the $t$-entropy, which does not have any obvious relations to any of the popularly known entropies. We analytically show that the proposed measure satisfies the major axiomatic properties of an entropy. The mathematical properties of the $t$-entropy were also rigorously analyzed. The efficacy of the $t$-entropy is demonstrated on a suit of application including image processing, divergence-based robust point estimation and subspace clustering. The consistency and robustness properties of the $t$-entropy based estimators are theoretically explored. In particular, we show that under standard regularity conditions, the estimator has an asymptotic break-down point of $0.5$, which is desired for robustness against outliers.  One possible extension of our method could be to extend the proposed entropy by using a general class of bounded concave functions. 
The application of $t$-entropy in the context of Power $k$-means clustering 
and sparse signal recovery 
are also some possible avenues for future research.

\bibliographystyle{apalike}
\bibliography{mybib}

\appendix
\section*{Appendices}
\section{Proofs from section \ref{property}}\label{a1}
In this section, we discuss the proofs of the properties of $t$-entropy. Before we proceed, let us first prove the following lemmas.
\label{S3}
\paragraph{Lemma \ref{l:1}}
The function $f(x)=\tan^{-1}\left(\frac{1}{x}\right)$ is convex on $[0,\infty)$.
\begin{proof}
We have, $f'(x) = - \frac{1}{1+x^2}$ and thus, $f''(x) =  \frac{2x}{(1+x^2)^2} \geq 0$ for all $x \geq 0$. Thus $f(x)=\tan^{-1}\bigg(\frac{1}{x}\bigg)$ is convex on $[0,\infty)$.
\end{proof}
\paragraph{Lemma \ref{l:2}}
For any $c>0$, the function $f(x)=x \tan^{-1}(\frac{1}{x^c})$ is concave on $[0,1]$. 
\begin{proof}
By some easy algebra, we have, 
\begin{align*}
    f'(x) & = \tan^{-1}\bigg(\frac{1}{x^c}\bigg)-\frac{c x^c}{1+x^{2c}}
\end{align*}
and
\begin{align*}
    f''(x) & = - \frac{c x^{c-1}}{(1+x^{2c})^2} \times (1+x^2c+c(1-x^c)).
\end{align*}
If $x \in [0,1]$, we have $f''(x) \leq 0$. Hence the result.
\end{proof}

\paragraph{Lemma \ref{lem3}}
Convex combination of finite number of concave functions is a concave function. 
\begin{proof}
Let, $g:\mathbb{R}\rightarrow \mathbb{R}$ be a convex combination of $n$ concave functions $f_i:\mathbb{R}\rightarrow \mathbb{R} \forall i\in \{1,\dots,n\}$, i.e., let $g(x)=\sum_{i=1}^n\alpha_if_i(x)$, $0\leq \alpha_i \leq 1$, $\sum_i \alpha_i=1$. Then for $0\leq \lambda \leq 1$, $g(\lambda x+(1-\lambda)y)=\sum_{i=1}^n\alpha_if_i(\lambda x+(1-\lambda)y)\geq \sum_{i=1}^n\alpha_i\bigg(\lambda f_i(x)+(1-\lambda)f_i(y)\bigg)=\lambda \sum_{i=1}^n\alpha_if_i(x)+(1-\lambda)\sum_{i=1}^n\alpha_if_i(y)=\lambda g(x)+(1-\lambda)g(y)$. Thus $g$ is a concave function.
\end{proof}

\section{Example from Binomial Distribution}
In this example, we consider the estimation for a binomial model. Let $X_1,\dots,X_n$ be i.i.d. from $Binomial(N,\theta)$, where $N$ is known. Our goal is to estimate $\theta$. For our experiment, we take $N=100$, $n=200$ and $\theta=0.2$. The large sample size of $n=100$ should help us estimate $\theta$ with a good precision. To make the problem more difficult, we deliberately add 10 outliers, drawn from the set $\{91,92,\dots,100\}$. Let $\hat{\theta}_{mle}$ be the M.L.E. of $\theta$ and $\hat{\theta}_{t}$ be the estimate obtained by applying Eqn ($5$). We generate 100 such datasets and for each of them, we compute $\hat{\theta}_{mle}$ and $\hat{\theta}_{t}$ and plot the obtained histograms in Fig. \ref{fig:binom}. The kernel density estimates for the distribution of both $\hat{\theta}_{mle}$ and $\hat{\theta}_{t}$ are shown in Fig. \ref{fig:binom}. It can be easily observed from Fig. \ref{fig:binom} that $\hat{\theta}_t$ is concentrated around the true parameter value of $\theta=0.2$, whereas $\hat{\theta}_{mle}$ consistently overestimates $\theta$.
\begin{figure}[ht]
    \centering
    \includegraphics[height=0.4\textwidth,width=0.4\textwidth]{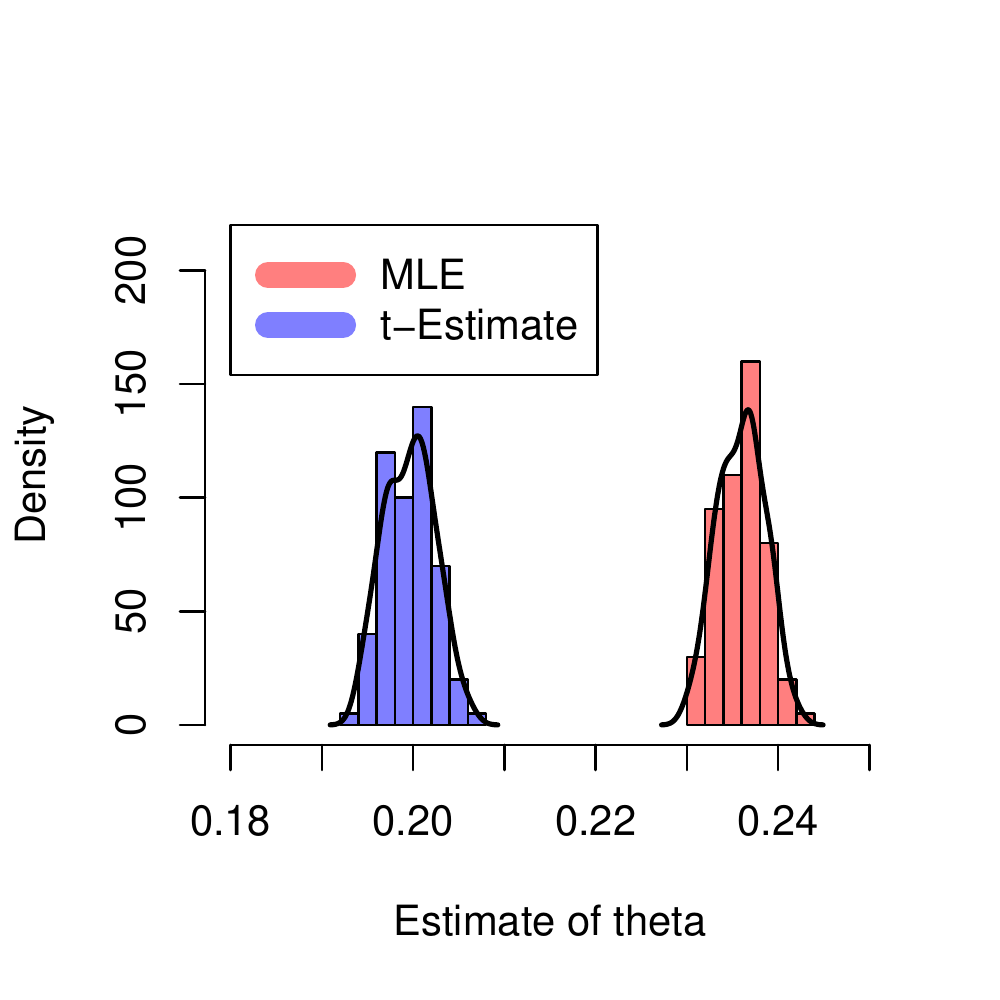}
    \caption{Histogram of $t$-estimates and MLE's for $\theta=0.2$ with $5\%$ contamination, showing that the $t$-estimate is more robust than MLE.}
    \label{fig:binom}
\end{figure}

\end{document}